\documentclass[11pt]{article}
\pdfoutput=1
\usepackage{amsmath,amssymb,amsthm,color,graphicx}
\usepackage{thmtools}
\usepackage{thm-restate}
\usepackage{lineno}
\usepackage[subrefformat=parens,labelformat=parens]{subcaption}
\usepackage{booktabs}
\usepackage[top=1in, bottom=1in, left=1in, right=1in]{geometry}
\usepackage{authblk}
\graphicspath{{figures/}}

\begin{document}
\newcommand{\modif}[1]{\textcolor{red}{\textbf{#1}}}

\newtheorem{theorem}{Theorem}
\newtheorem{corollary}{Corollary}
\newtheorem{definition}{Definition}
\newtheorem{lemma}{Lemma}
\newtheorem{fact}{Fact}
\newtheorem{proposition}{Proposition}
\newtheorem{example}{Example}
\newtheorem{remark}[theorem]{Remark}
\newtheorem{conjecture}[theorem]{Conjecture}

\newcommand{\MC}{\mathcal{M}}
\newcommand{\MCadj}{\mathcal{M}_{adj}}
\newcommand{\MCany}{\mathcal{M}_{any}}
\newcommand{\MCanyrev}{\mathcal{M}_{any}^*}
\newcommand{\bigO}{\mathcal{O}}
\newcommand{\Sn}{\mathcal{S}_n}
\newcommand{\mz}{\mathit{mz}}
\newcommand{\MZ}{\mathit{MZ}}
\newcommand{\perm}{\mathit{perm}}
\newcommand{\MCblocks}{\mathcal{M}_{\mathit{blocks}}}
\newcommand{\myqed}{\hfill$\square$}

\def\polylog{\operatorname{polylog}}

\pagestyle{headings}  % switches on printing of running heads

\title{On computing the total displacement number via weighted Motzkin paths}

\author[1]{Andreas B\"artschi\thanks{andreas.baertschi@inf.ethz.ch}}
\author[1]{Barbara Geissmann\thanks{barbara.geissmann@inf.ethz.ch}}
\author[1]{Daniel Graf\thanks{daniel.graf@inf.ethz.ch}}
\author[1]{Tomas Hruz\thanks{tomas.hruz@inf.ethz.ch}}
\author[1]{Paolo Penna\thanks{paolo.penna@inf.ethz.ch}}
\author[1]{Thomas Tschager\thanks{thomas.tschager@inf.ethz.ch}}
\affil[1]{ETH Z\"urich, Department of Computer Science}

\def\REPOURL{\url{http://people.inf.ethz.ch/grafdan/motzkin/}}

\date{}
\maketitle
%
%\linenumbers

\begin{abstract}
	Counting the number of permutations of a given total displacement is equivalent to counting weighted Motzkin paths of a given area (Guay-Paquet and Petersen~\cite{guay2014generating}). The former combinatorial problem is still open. 
	In this work we show that this connection allows to construct efficient algorithms for counting and for sampling such permutations.
	These algorithms provide a tool to better understand the original combinatorial problem.
	A by-product of our approach is a different way of counting based on certain ``building sequences'' for Motzkin paths, which may be of independent interest. 

\end{abstract}

\section{Introduction}
Consider the set $\Sn$ of all permutations over $n$ elements $\{1,2,\ldots,n\}$. Diaconis and Graham~\cite{diaconis1977spearman} studied the \emph{disarray} statistic of permutations, also called \emph{total displacement} by Knuth~\cite[Problem 5.1.1.28]{donald1999art}, defined as follows. For any permutation $\pi$ define its distance to the identity permutation as the sum of the displacements of all elements:
\[
D(\pi):= \sum_{i=1}^n |i - \pi(i)| = 2\sum_{\pi(i)> i} (\pi(i)-i).
\]
Note that this distance is always \emph{even}. The following natural question is still unresolved:
\begin{quote}
	\emph{How many permutations at a given distance $2d$ from the identity permutation are there?}
\end{quote}	
That is, one would like to know the following \emph{total displacement} number: 
\[
D(n,d):=|\{\pi \in \Sn \mid D(\pi)=2d\}|,
\]
that is the number of permutations of total displacement equal to $2d$.
So far, a closed formula for arbitrary $n$ is only known for fixed $d$ \emph{up to seven} ($d\leq 7$)~\cite{guay2014generating}.
Entry A062869 \cite{oesis_triangle}
of the OEIS reports values of $D(n,d)$ for small $n$ and $d$ ($n\leq 30$).

Guay-Paquet and Petersen~\cite{guay2014generating} made recently significant progress in this question by showing that these permutations are in correspondence to \emph{Motzkin paths} whose \emph{area} is exactly the distance $d$ under consideration. 
Their result shows that, for any Motzkin path (see below) of area $d$, one can easily calculate the number of permutations that correspond to this specific path. Therefore the problem above translates into the problem of counting \emph{weighted} Motzkin paths of a given area. 

A Motzkin path consists of a sequence of $U$ (Up-right), $H$ (Horizontal-right), and $D$ (Down-right) moves over the two-dimensional lattice starting at coordinate $(0,0)$ and such that the path never goes below the  $y=0$ axis and ends on the $y=0$ axis (see Figure~\ref{fig:Motzkin-and-permutation} (right) for an example). For any such path, one can consider its \emph{width} and its \emph{area} defined as the number of moves and the area of the region between the $y=0$ axis and the path.
\begin{figure}[t]
	\centering
	\includegraphics{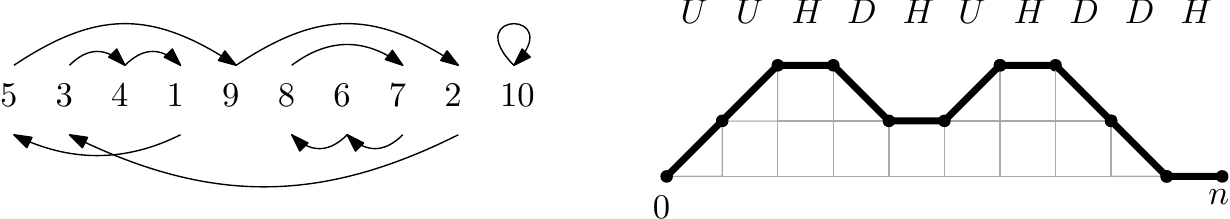}
	\caption{A permutation and its Motzkin path of width $10$ and area $12$.}
	\label{fig:Motzkin-and-permutation}
\end{figure}
The permutations over $n$ elements with total displacement $2d$ map into Motzkin paths of width $n$ and area $A=d$. 

For instance, the permutation in Figure~\ref{fig:Motzkin-and-permutation} is mapped into a Motzkin path according to the following rule. The first element $\pi(1)=5$ is mapped into a $U$ because the element at position $1$ goes to a higher position (right) and also the number coming into position $1$ is higher than $1$: $\pi(1)>1<\pi^{-1}(1)$. The fourth element is mapped into $D$ because the opposite happens: $\pi(4)=1<4>3=\pi^{-1}(4)$. Finally, elements $3,5,7,10$ are mapped into $H$ because neither of the previous cases apply. 

Let $h_i$ denote the maximum height of the path during move $i$ (for $U$: after the move, for $D$: before the move, and anytime for $H$).
Then the number $\omega(\mz)$ of permutations that map to a certain Motzkin path $\mz$ is~\cite{guay2014generating}
\begin{eqnarray}
\omega(\mz) = \prod_i^n\omega_i \text{ where } \omega_i = 
\begin{cases} h_i & \text{if } \mz_i=U \text{ or } \mz_i = D,\\ 2h_i+1 & \text{if } \mz_i = H. \label{eq:perm-product-introduction}
\end{cases}
\end{eqnarray}
We also refer to $\omega(\mz)$ as the \emph{weight} of $\mz$.
In the example in Figure~\ref{fig:Motzkin-and-permutation} this gives $1\cdot 2 \cdot 5 \cdot 2 \cdot 3 \cdot 2 \cdot 5 \cdot 2 \cdot 1 \cdot 1 = 1200$.
Note how this formula separates over the moves of the Motzkin path. This independence is what we will exploit in this article.

\begin{theorem}[\cite{guay2014generating}]\label{th:displacement_equals_weighted_paths}
	For any $n$ and $d$, let $\MZ(n,A)$ be the set of all Motzkin paths of width $n$ and area $A=d$. Then it holds that
	\begin{equation}\label{eq:displacement_formula_via_paths}
	D(n,d) = \sum_{\mz \in \MZ(n,A)} \omega(\mz).
	\end{equation}
\end{theorem}

\begin{restatable}[Appendix~\ref{app:omitted-proofs}]{corollary}{pathtopermcorollary}
\label{cor:path-to-permutation}
Given a Motzkin path $\mz$ of length $n$, we can sample uniformly at random one of the $\omega(\mz)$ many permutations mapping into $\mz$ in time $\bigO(n)$.
\end{restatable}

\paragraph{Our contribution.}
In this work, we address counting and sampling of permutations from both a combinatorial and computational point of view. Specifically:
\begin{itemize}
	\item On the computational side, we show that the total
          displacement number $D(n,d)$ can be computed efficiently,
          namely, in time $\bigO(n^4)$ and $\bigO(n^3)$ space.
	\item On the combinatorial side, we introduce sequences of
          certain \emph{building blocks} which provide a different
          perspective on the problem structure. Moreover, this is a
          crucial part of a Markov chain sampling method which constitutes the
          third contribution of this paper.
	\item Finally, we consider the task of \emph{sampling}
          permutations of a given total displacement with uniform
          distribution.
\end{itemize}

To compute the number of permutations efficiently, we look at the
paths from left to right. Building on an operation introduced by
Barcucci et al. \cite{barcucci1995construction}, we can provide an
elegant dynamic programming formulation which achieves a running time
of $\bigO(n^4)$ and needs space $\bigO(n^3)$. Consequently, we can
compute the sequences A062869 \cite{oesis_triangle} and A129181
\cite{oesis_motzkin_area} to much higher values of $n$ and $d$ than was
possible before.

Considering the combinatorial aspects, we show that every Motzkin path
comes from a sequence $a$ describing its building blocks. We provide an
explicit formula for the number $m(a)$ of paths that these building
blocks can form. The weights in Equation~\eqref{eq:perm-product-introduction}
are preserved in the sense that the weight of a path depends only on
its building sequence.

Since the exact formula seems to be currently out of reach, to
achieve good estimates of $D(n,d)$ for very large $n$ and $d$, we
contribute sampling methods which can also be of independent interest.
In particular, the dynamic programming algorithm provides a sampler
with the same complexity as the algorithm itself.  Further, we show
that sampling sequences of building blocks with appropriate
distribution automatically gives a sampler for the permutations. One
application of the latter result is a Monte Carlo Markov chain (MCMC)
method which gives an alternative approach to the dynamic
programming. The computational experiments with the MCMC method show a
promising convergence speed leading to a sampler with very high values
of $n$ and $d$. The experimental results support a hypothesis that the
MCMC method is faster than the method based on dynamic programming and
runs in $\bigO(n^3)$ time.

\paragraph{Related Work.}
Different metrics on permutations have been studied, for a
survey see \cite{Deza98metricsonpermutations}.  Sampling and counting
of permutations of a fixed distance was studied for several
metrics~\cite{samplingpermutations} but not for total
displacement.

The number of Motzkin paths under various conditions were also studied in
a more general frame of enumeration of lattice paths
\cite{humphreys2010history,goulden2004combinatorial}. Motzkin
numbers play a role in many combinatorial problems as is illustrated
for example in \cite{donaghey1977motzkin}.
The total area under a set of generalized Motzkin paths, where the
horizontal segments have a constant length $k$ ($k\geq 0$) have been
studied in \cite{pergola2002bijective} and
\cite{merlini2003generating}. Moreover, the author in
\cite{sulanke2000moments} studies the moments of generalized Motzkin
paths where the first moment describes the area under a Motzkin
path.
Heinz~\cite{oesis_motzkin_area} describes a different
algorithm for enumerating unweighted Motzkin paths with a given area,
cf.~Remark~\ref{rem:heinz} in Section~\ref{sec:DP-counting}.

The use of Markov chains for sampling and counting combinatorial objects is 
a very active research area (see e.g. the book \cite{Bub11}), and some 
works exploit the connection between combinatorial structures and paths
of a certain type to accomplish this task (see e.g. \cite{GrePasRan09}).

\paragraph{Paper Organization.} Section~\ref{sec:DP}
describes the dynamic programming
algorithm. Section~\ref{sec:blocks_decomposition} describes how
weighted Motzkin paths can be counted via building block sequences.
Section~\ref{sec:sampling} provides a Markov chain sampling algorithm as well as its experimental evaluation.

% Main Content

\section{Weighted Motzkin Paths using Dynamic Programming}
\label{sec:DP}

Recall that we denote by $D(n,d)$ the number of permutations on $n$ elements with total displacement $2d$ (OEIS A062869 \cite{oesis_triangle}).
Let $M(n,A)$ denote the number of Motzkin paths of width $n$ and area $A$ (OEIS A129181 \cite{oesis_motzkin_area}).

\subsection{Dynamic Program for Counting Weighted Motzkin Paths}
\label{sec:DP-counting}
\begin{theorem}
Computing $M(n,A)$ and $D(n,d)$ can be done in time $\bigO(n^4)$ and space $\bigO(n^3)$.
\label{thm:last_fall_dp}
\end{theorem}

\begin{proof}

The key ingredient is a construction by Barcucci et al.~\cite{barcucci1995construction} that produces every possible Motzkin path through a unique sequence of insertion steps. 

Let us look at the last fall of a given Motzkin path, i.e., its suffix of Down-right moves.
At one of the positions before or after any of these fall moves, we insert a new \emph{peak} (a $U$ and a $D$) or we insert a new \emph{flat} (an $H$). 
Repeatedly inserting peaks and flats this way along the last fall will create our path. See Figure~\ref{fig:last-fall-step} 
for an example. 
\begin{figure}[b]
	\centering
	\includegraphics[width=\textwidth]{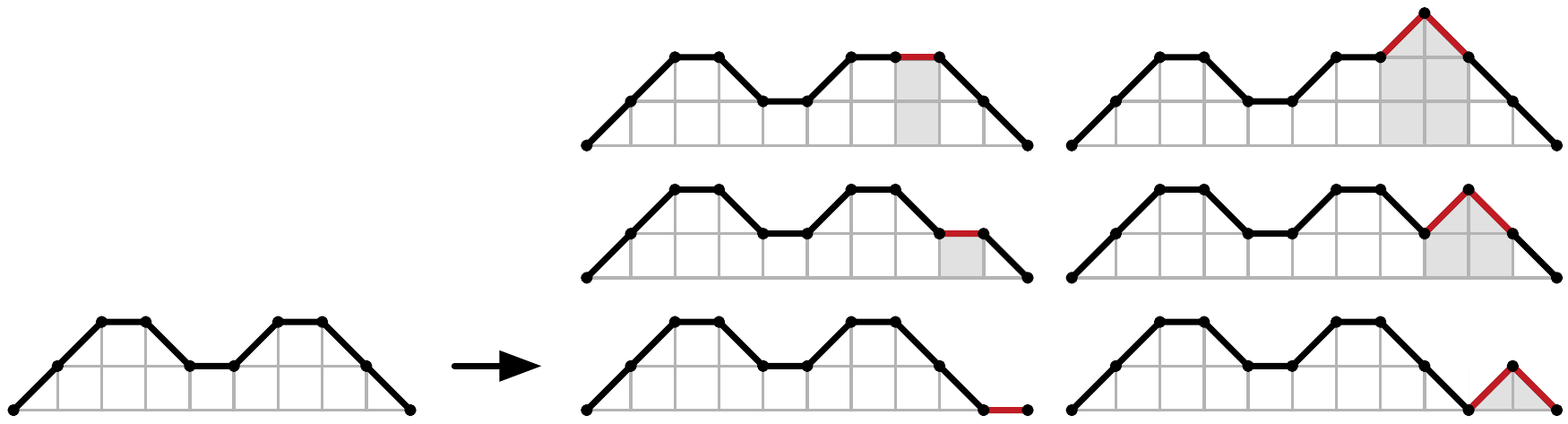}
	\caption{All six possible flat- and peak-extensions of the last fall of length $2$.}
	\label{fig:last-fall-step}
\end{figure}
This construction is complete and unique~\cite{barcucci1995construction}, meaning that every Motzkin path can be created through a unique sequence of such insertions.

This allows us to derive a dynamic programming formulation for counting $M(n,A)$.
We add the last fall length $l$ to our state and write it as $M(n,A,l)$.
So how can we recursively express $M(n,A,l)$? We undo the last insertion step. If we inserted a flat last, then we were at $M(n-1,A-l,l')$ before the insertion, for some $l'\geq l$, because the last fall was at least as long before the insert. When inserting a peak, we might increase the last fall length by one, but not more. So $M(n-2,A-(2l-1),l')$ for all $l' \geq l-1$ are also possible predecessor states. Together with the base case $M(0,0,0)=1$ this gives the recurrence
\begin{linenomath}
\begin{align}
M(n,A,l) = 
&\sum_{l' \geq l}^{{n}/{2}} M(n-1,A-l,l') +
\sum_{l' \geq l-1}^{{n}/{2}} M(n-2,A-(2l-1),l'),
\label{eq:Motzkin-area-recurrence}
\end{align}
\end{linenomath}
which allows for $\bigO(n^4)$ many states as $A \leq n^2$ and $l \leq n$. Hence we immediately get an $\bigO(n^5)$ time algorithm with $\bigO(n^4)$ space. We can shave off one factor of $n$ in both time and space as follows:
We first note, that we can compute the two sums in constant time if we precompute the prefix sums over the last variable $l'$.
Let us denote these prefix sums as $\mathit{SM}(n,A,l) = \sum_{l'=0}^l M(n,A,l') = \mathit{SM}(n,A,l-1) + M(n,A,l)$.
This allows us to compute every value of $M(n,A,l)$ in amortized constant time, so in time $\bigO(n^4)$ overall.
Finally, our recurrence only relies on the last two values of $n$, so when computing $M(n,\cdot, \cdot)$ only the $\bigO(n^3)$ many values for $M(n-1,\cdot,\cdot)$ and $M(n-2,\cdot,\cdot)$ need to be stored.
The values $M(n,A)$ are then simply the marginals of $M(n,A,l)$ over all last fall lengths $l$.

We can extend this recurrence to the weighted case which by Corollary~\ref{cor:path-to-permutation} gives rise to the total displacement count:
We distribute the factors of the weight $\omega(\mz)$ (Equation \eqref{eq:perm-product-introduction}) over the steps of the dynamic program. As $l$ denotes the height of the last flat or peak that we add, we have factors $2l+1$ or $l^2$: 
\begin{linenomath}
\begin{align*}
\ \,
D(n,d,l) = 
&(2l+1)\sum_{l' \geq l}^{{n}/{2}} D(n-1,d-l,l')+
l^2\sum_{l' \geq l-1}^{{n}/{2}} D(n-2,d-(2l-1),l'). \qedhere
\ \,
\end{align*}
\end{linenomath}
\end{proof}

\begin{remark}
The bounds in Theorem~\ref{thm:last_fall_dp} assumed that basic operations have unit-cost.
The numbers involved can be exponential in $n$ however.
We can easily bound $M(n,A) \leq 3^n$ and $D(n,d) \leq n!$ showing that their bit-representations are at most of length $\bigO(n\log n)$.
Our dynamic programs only use multiplication with small numbers of size $\bigO(\log n)$ and addition.
So one can consider a refined analysis by multiplying both the time and space bounds of Theorem~\ref{thm:last_fall_dp} by $\Theta(n \polylog n)$. Finally, as suggested by an anonymous reviewer, the space could be further improved by counting modulo small primes and using the Chinese Reminder Theorem.
\end{remark}

\begin{remark}
\label{rem:heinz}
For computing $M(n,A)$, the OEIS contains a dynamic program by Heinz~\cite{oesis_motzkin_area}.
It is stated as a \emph{Maple} code snippet without any further comment or reference.
It uses a different state and might have the same time complexity as ours.
We believe that our extension to the weights of $D(n,d)$ can also be applied.
\end{remark}

\subsection{Sampling from the Dynamic Program}
\label{sec:DP-sampling}

\begin{theorem}
After running the dynamic program from Theorem~\ref{thm:last_fall_dp}, we can sample (weighted) Motzkin paths in time $\bigO(n)$.
\label{thm:last_fall_sampling}
\end{theorem}

\begin{proof}
Given access to a source of randomness and the filled table for $M$, we can randomly retrace the steps through the dynamic programming states to sample a Motzkin path from right to left.
For the weighted paths according to $D(n,d)$ all the steps will be exactly the same.
We first sample the last fall length by picking a random number $x \in_{\text{u.a.r.}} \{0, \dots, M(n,A)-1\}$ and then finding the smallest $l$ such that its prefix sum $\mathit{SM}(n,A,l)$ is larger than $x$.
We continue with $x-\mathit{SM}(n,A,l-1)$, the offset within the class of paths with last fall length $l$.
For each step, we first decide whether we are in the flat-case or in the peak-case of the recurrence by comparing $x$ to the left summand of~\eqref{eq:Motzkin-area-recurrence}.
We then know whether the move before the last fall was an $H$ or a $U$. We increment $l'$ until we find the last fall length of the previous state.
We adapt $x$ and recurse until we end at $M(0,0,0)$ with $x=0$.
Note that the search for the initial $l$ takes linear time. After that, every time we compare $x$ to a value of $M$, we fix at least one move of the sampled Motzkin path, so sampling takes $\bigO(n)$ time overall.
\end{proof}

\begin{remark}
This sampling procedure requires the full table of the dynamic program to be stored. Hence the memory optimization from $\bigO(n^4)$ to $\bigO(n^3)$ in Theorem~\ref{thm:last_fall_dp} can not be used simultaneously.
\end{remark}

\begin{remark}
A C++ implementation of our counting and sampling approaches by Theorems~\ref{thm:last_fall_dp} and~\ref{thm:last_fall_sampling} is available at \REPOURL. With our code, we can quickly compute for $n$ up to $100$ (and all $d$) the integer sequences A062869 \cite{oesis_triangle} and A129181 \cite{oesis_motzkin_area} which were only known up to $n\leq 30$ and $n\leq 50$ before.
\end{remark}

\section{Combinatorial Structure of Motzkin Paths}\label{sec:blocks_decomposition}

In this section, we look at the combinatorial structure of Motzkin paths:
There is a natural decomposition of any Motzkin path into ``building blocks'', already hinted at in the last section. 
For each height $i$ of the Motzkin path we count the number of flats $f_i$ and peaks $p_i$.

\begin{definition}[building sequence]\label{def:feasible_sequence}
	For given positive integers $n$ and $A$, 
	a finite sequence of non-negative integers $a=(f_0,p_1,f_1,p_2,\ldots,p_h,f_h)$ is a \emph{building sequence} if all $p$-entries are \emph{non-zero},
	$
	p_1,p_2,\ldots,p_h >0,
	$
	and the following two conditions hold:
	\begin{eqnarray}
	(f_0 + f_1 + \ldots + f_h) + 2(p_1 + p_2 + \ldots + p_h)		= n,	\label{eq:widthI} \\		
	(0f_0 + 1 f_1 + \ldots + h f_h) + (1p_1 + 3p_2 + \ldots + (2h-1)p_h)	= A.	\label{eq:areaI}
	\end{eqnarray}
	The set of all building sequences satisfying \eqref{eq:widthI}-\eqref{eq:areaI} is denoted as $S(n,A)$.
\end{definition}

\begin{figure}[t!]
	\centering
	\includegraphics{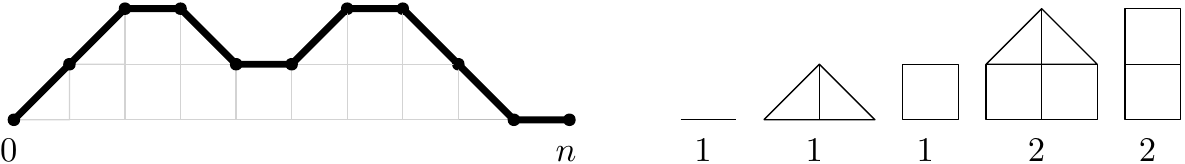}
	\caption{The Motzkin path on the left can be obtained from its building blocks.}
	\label{fig:building_blocks}
\end{figure}

Such a sequence has a natural interpretation as a set of ``building blocks'' that generate a number of Motzkin paths of width $n$ and area $A$ (see Figure~\ref{fig:building_blocks}):
We have $f_i$ flats and $p_i$ peaks of height $i$ which can be split into pieces of width~$1$ and then rearranged into a Motzkin path.  

\begin{proposition}\label{prop:unique_sequence}
	For any Motzkin path $\mz$ of width $n$ and area $A$ there exists a unique building sequence $a^{(\mz)}\in S(n,A)$ 
	such that $\mz$ can be obtained by splitting and rearranging the blocks of this sequence.
\end{proposition}

Theorem~\ref{th:displacement_equals_weighted_paths} gives a surjective mapping from permutations into Motzkin paths.
It is easy to see that the number of permutations $\omega(\mz)$ mapping into the same path $\mz$, given by Equation~\eqref{eq:perm-product-introduction}, is uniquely determined by the building block sequence $a = a^{(\mz)}$, since we have 
\begin{eqnarray}
	\perm(a) := \prod_{f_i}{(2i+1)^{f_i}} \prod_{p_i}{i^{2p_i}} = \omega(\mz).	\label{eq:perm-product}
\end{eqnarray}

Hence $\omega(\mz)$ is independent of the actual Motzkin path and only depends on its combinatorial structure.
This raises the question of whether also the number of Motzkin paths which share a common building sequence $a$ is solely determined by $a$.
We answer this in the positive, deriving a formula for this number, denoted by $m(a)$. 
We proceed in a top-down fashion by looking at the number of peaks and flats in the highest level and how these can be rearranged. 
Once a level is fixed, we proceed recursively by arranging the blocks one level below.

\begin{theorem}\label{th:blocks_rearrange}
	For any building sequence $a= 
	(f_0, p_1, f_1, \ldots, p_h, f_h) \in S(n,A)$, 
	the number of Motzkin paths of width $n$ and area $A$ that can be constructed out of the building sequence $a$ is exactly
	\begin{linenomath}
	\begin{align}	
			m(a) 	&=		{f_h + p_h -1 \choose p_h -1} {p_{h} + f_{h-1} \choose f_{h-1}} 
			{p_{h} + f_{h-1} + p_{h-1} - 1 \choose p_{h-1}-1} {p_{h-1} + f_{h-2} \choose f_{h-2}} \cdots 	\nonumber\\
			&\quad\ \cdots 
			{p_{3} + f_{2} + p_{2} - 1 \choose p_{2}-1}  {p_2 + f_1 \choose f_1}
			{p_{2} + f_{1} + p_{1} - 1 \choose p_{1}-1}  {p_1 + f_0 \choose f_0}_.	\label{eq:ma-product}
	\end{align}
	\end{linenomath}
\end{theorem}

\begin{figure}[t!]
	\centering
	\includegraphics[width=\linewidth]{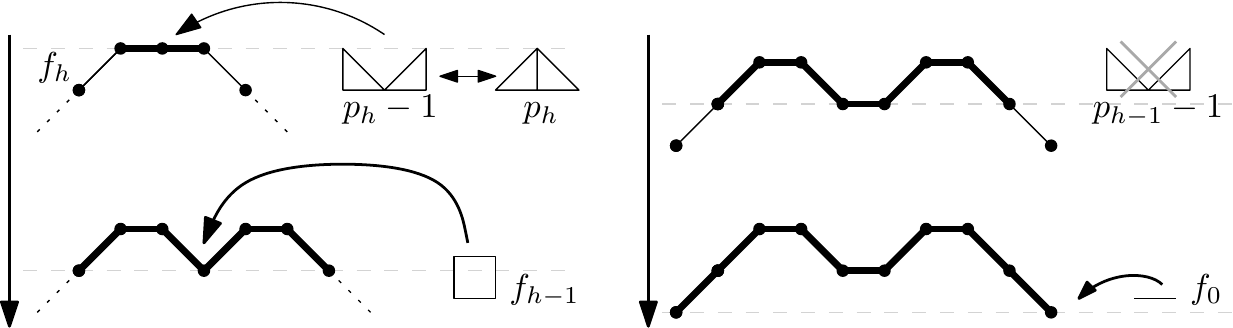}
	\caption{The top down construction of paths from the given sequence $(1,1,1,2,2)$.
		Note that $p_1-1=0$, and thus no $\mathit{DU}$ valley is inserted at height $1$.
	}
	\label{fig:top-down-building}
\end{figure}

\begin{proof}
	We start with the highest flats of the sequence $a$. There are $f_h$ of those flats. 
	Two (or more) such flats can either lie directly next to each other, 
	or they might be separated by a \emph{Down-right} move followed at some point by an \emph{Up-right} move. 
	We call this setting a $\mathit{DU}$ valley; we get such valleys by splitting peaks of height $h$ and reassembling them the other way round, 
	see Figure~\ref{fig:top-down-building}.
	A feasible Motzkin path has to have a $U$ slope at the very left and a $D$ slope at the very right of all height $h$ pieces. 
	The remaining $p_{h}-1$ $\mathit{DU}$ valleys can be freely placed around the $f_h$ flats, 
	that is we choose their places from $f_h + p_h -1$ available positions. 
	The number of ways to do this is
	\begin{linenomath}
	\begin{align}
		\label{eq:peak_combos}
		{f_h + p_h -1 \choose p_h -1}_.
	\end{align}
	\end{linenomath}
	Now we continue on the second highest level $h-1$. Naturally, the number of times that our Motzkin path rises above level $h-1$ is exactly the number $p_h$
	of peaks of height $h$. We can distribute our $f_{h-1}$ flats of height $h-1$ around those peaks, i.e. pick from $p_h + f_{h-1}$ many positions,
	hence we can choose from 
	\begin{linenomath}
	\begin{align}
		\label{eq:flat_combos}
		{p_{h} + f_{h-1} \choose f_{h-1}}
	\end{align}
	\end{linenomath}
	many possibilities. After placing the flats, we will have to place new valleys down to the next lower level around the existing $p_h$ peaks and
	$f_{h-1}$ flats. As before, the leftmost up and down slopes are fixed, hence the number of ways to distribute $p_{h-1} - 1$ valleys is
	given by the third factor in Equation~\eqref{eq:ma-product}.
	Since the choices in different levels are independent, we can iterate this reasoning until we include flats of height $0$.
\end{proof}

We conclude with a corollary of Theorems~\ref{th:displacement_equals_weighted_paths}~and~\ref{th:blocks_rearrange}: 

\begin{corollary}\label{cor:mapping}
	There exists a surjective mapping from permutations over $n$ elements into \emph{building sequences} satisfying the following condition:
	For any building sequence $a\in S(n,A)$, the number of permutations $\pi$ which are at distance $D(\pi) = 2d = 2A$ from the identity permutation 
	and that are mapped into this building sequence $a$ is precisely
	\begin{equation}\label{eq:perm-mapping}
	P(a):=m(a)\cdot perm(a),		
	\end{equation}
	where $m(a)$ is given by Equation~\eqref{eq:ma-product} and $\perm(a)$ by Equation~\eqref{eq:perm-product}.
	Therefore the total number of permutations at distance $2d=2A$ from the identity permutation satisfies
	\begin{equation}\label{eq:perm-count-via-sequences}
		D(n,d) \stackrel{\eqref{eq:displacement_formula_via_paths}}{=}
 		\sum_{\mz \in \MZ(n,A)}{\omega(\mz)} =	\sum_{a \in S(n,A)}{ P(a)}.
	\end{equation}
\end{corollary}

\begin{example}
	The building blocks in Figure~\ref{fig:building_blocks}  yield ${3\choose 1}{3\choose 1}{3\choose 0}{2\choose 1} =18$ Motzkin paths, 
	and each path corresponds to $1200$ permutations. 
	So, there are $1200\cdot 18=21\, 600$ permutations mapping into the building sequence $a=(1,1,1,2,2)$.
\end{example}

\begin{remark}
	Theorem~\ref{th:blocks_rearrange} and Corollary~\ref{cor:mapping} allow for a dynamic program for counting and sampling weighted Motzkin paths,
	similar to Sections~\ref{sec:DP-counting} and~\ref{sec:DP-sampling}.
	Additionally, we can easily sample paths with a fixed number of highest peaks and flats, at the cost of an additional $O(n^3)$-factor in the running time,
	see Appendix~\ref{app:DP-top-down}.
\end{remark}

\section{Sampling Weighted Motzkin Paths by Length and Area}\label{sec:sampling}
In this section, we consider the task of selecting (sampling) permutations with uniform distribution over all permutations of a given total displacement.
By Corollary~\ref{cor:path-to-permutation} it is enough to sample Motzkin paths with the proper weights.
We have already seen in Section~\ref{sec:DP-sampling} that we can sample such weighted Motzkin paths using dynamic programming at the cost of large memory consumption.

We will show in Section~\ref{sec:sampling-markov-chain} an approach to sample weighted Motzkin paths based on the building sequences introduced in Section~\ref{sec:blocks_decomposition} that requires only $\bigO(n)$ memory.
In general, observe that sampling permutations
can be accomplished efficiently if we can sample  building sequences with a probability proportional to  $P(a)=m(a)\cdot perm(a)$ in polynomial time:  

\begin{restatable}[]{theorem}{sequencesamplingtopermutationstheorem}
\label{th:sequence-sampling-to-permutations}
	Every polynomial-time algorithm that samples sequences in $S(n,A)$ with probability $\pi(a)\propto P(a)$ can be turned into a polynomial-time algorithm for sampling permutations  uniformly at random among the permutations over $n$ elements and of total displacement $2d=2A$. 
\end{restatable}
\begin{proof}
Given a sequence $a\in S(n,A)$, the sampler maps this sequence into a random Motzkin path, and then into a random permutation as follows:
\begin{enumerate}
	\item	\label{sampler:seq-to-motz} Pick a Motzkin path $mz$ uniformly at random among those that can be created with $a$, that is, with
	probability
	$\frac{1}{m(a)}$. 
	\item	\label{sampler:motz-to-perm} Pick a permutation u.a.r.\ among those that map into the Motzkin path $mz$, that is, with probability
	$\frac{1}{perm(a)}$.
\end{enumerate}

\paragraph{Step~\ref{sampler:seq-to-motz} (sequences to Motzkin paths).} The top-down construction used to prove Theorem~\ref{th:blocks_rearrange} suggests also how to sample one of the $m(a)$ Motzkin paths for a given sequence $a$ with uniform distribution. Namely, we pick the positions of the $DU$ valleys at height $h$ uniformly at random (Equation~\eqref{eq:peak_combos}), then we pick the positions of the  $f_{h-1}$ flats uniformly at random (Equation~\eqref{eq:flat_combos}), and repeat this to the lower level exactly as described in the top-down construction. Since a particular path corresponds to exactly one choice in each of these steps, by Equation~\eqref{eq:ma-product} its probability is precisely $1/m(a)$.

\paragraph{Step~\ref{sampler:motz-to-perm} (paths to permutations).} This is shown in Corollary~\ref{cor:path-to-permutation} above.
\end{proof}

\subsection{Preliminary definitions on Markov chains}
In this section, we introduce some of the definitions on Markov chains used throughout this work (see e.g. \cite{LevPerWil09}). 
A Markov chain over a finite state space $S$ is specified by a transition matrix $P$, where $P(a,a')$ is the probability of 
moving from state $a$ to state $a'$ in one step. The $t^{th}$ power of the transition matrix gives the probability of moving from one state to another state in $t$ steps. The chain studied in this work is ergodic (see below for a proof), meaning that it has a unique \emph{stationary distribution} $\pi$, that is,  $\lim_{t\rightarrow \infty} P^t(a,a')=\pi(a')$ for any two states $a$ and $a'$.  

\paragraph{Reversible chains.}
We shall use the definition of a \emph{reversible} Markov chain, also called \emph{detailed balanced condition}: If the transition matrix $P$ admits a vector $\pi$ such that $\pi(a) P(a,a') = \pi(a')P(a',a)$ for all $a$ and $a'$, then $\pi$ is the \emph{stationary distribution} of the chain with transitions $P$. 

\paragraph{Total Variation Distance and Mixing Time.}
The \emph{total variation} distance of two distributions $\mu$ and $\pi$ is 
\[ d_{TV}(\mu,\pi) := \frac{1}{2} \cdot \sum_{a\in S} \Big|\mu(a)-\pi(a)\Big|. \]
The \emph{mixing time} of an ergodic Markov chain with transition matrix $P$ is defined as 
\[
t_{mix} (\epsilon) = \min\left\{t : d_{TV}(P^t(a_0,\cdot),\pi)\leq \epsilon
 \ \text{ for all } a_0 \in S\right\} .
\]
It is common to also define the quantity $t_{mix}:=t_{mix}(1/4)$, which is  justified by the fact that, for any $\epsilon$, $t_{mix}(\epsilon) \leq \lceil\log_2 1/\epsilon\rceil \cdot  t_{mix}$. In our experiments, we shall evaluate the total variation distance for $\epsilon =0.05$ to get better estimate.

\subsection{A Markov Chain Sampler}
\label{sec:sampling-markov-chain}
Suppose we have a set of $k$ possible local changes transforming any sequence $a$ into another sequence $a'$ such that all sequences can be obtained
by applying a certain number of such operations. Then the following standard Metropolis chain samples sequences with the desired distribution:
\begin{enumerate}
	\item	\label{trans1} With probability $\frac{1}{2}$ do nothing. Otherwise, 
	\item	\label{trans2} Select one of the $k$ local operations u.a.r. If this operation cannot be applied to the current sequence $a$ (the new
		sequence is unfeasible) do nothing; Otherwise, let $a'$ be the sequence obtained from $a$ by applying this operation;
	\item	\label{trans3} Accept the operation transforming $a$ to $a'$ with probability
		\begin{equation}
		\label{eq:accept_probability}
			A(a,a') :=\min\left\{1,\frac{P(a')}{P(a)}\right\}=\min\left\{1,\frac{m(a')}{m(a)}\cdot \frac{perm(a')}{perm(a)}\right\},
		\end{equation}
		and do nothing with remaining probability $1 - A(a,a')$.
\end{enumerate}

\paragraph{Local operations over the sequences}
We define our Metropolis chain $\MCblocks$ through four types of operations: Peak to Flat (PF), Flat to Valley (FV), Flat to Flat (FF), and Peak into Valley (PV). We formally define them as:
\begin{linenomath}
 \begin{align*}
 	PF(i,j) :=& \begin{cases}
 		p_i	& \leftarrow\	p_i - 1		\\
 		f_{i-1}	& \leftarrow\	f_{i-1} + 2	\\
 		f_j	& \leftarrow\	f_j - 1	\\
 		f_{j+1}	& \leftarrow\	f_{j+1} + 1	
 	\end{cases}, & 
 	FV(i,j) :=& \begin{cases}
 		f_i	& \leftarrow\	f_i - 2		\\
 		p_i	& \leftarrow\	p_i + 1	\\
 		f_j	& \leftarrow\	f_j - 1	\\
 		f_{j+1}	& \leftarrow\	f_{j+1} + 1	
 	\end{cases} 
	%\label{eq:MC:PF-FV}
 	\\
 	FF(i,j) :=& \begin{cases}
 		f_i	& \leftarrow\	f_i - 1		\\
 		f_{i+1}	& \leftarrow\	f_{i+1} + 1	\\
 		f_j	& \leftarrow\	f_j - 1		\\
 		f_{j-1}	& \leftarrow\	f_{j-1} + 1	\\
 	\end{cases} &
 	PV(i,j) :=& \begin{cases}
 		p_i	& \leftarrow\	p_i - 1		\\
 		f_{i-1}	& \leftarrow\	f_{i-1} + 2	\\
 		p_j	& \leftarrow\	p_j - 1		\\
 		f_{j}	& \leftarrow\	f_{j} + 2	\\
 	\end{cases}
 \end{align*}
\end{linenomath}
Note that each type of operation applies to two indices $i$ and $j$, and we also implicitly consider the reversed operations which ``undo'' the changes. We now explain step~\ref{trans2} of the chain $\MCblocks$ in more detail: The Markov chain $\MCblocks$ picks two indices $i$ and $j$ at random, then picks one of the four operations above, and decides with probability $1/2$ whether to choose the operation or its reversed version.
As for step~\ref{trans3}, computing the transitional probability $A(a,a')$ can be done in constant time as only a few of the factors in Equations~\eqref{eq:perm-product} and~\eqref{eq:ma-product} change.

\begin{restatable}[]{theorem}{MCstationarytheorem}
 \label{th:MC:stationary}
 	The Markov chain $\MCblocks$ defined above is ergodic and its unique stationary distribution satisfies $\pi(a)\propto P(a)$ for every $a\in S(n,A)$.
 \end{restatable}

\begin{proof}
The proof consists of two steps. First, we have to show that the chain is ergodic, that is, it is aperiodic and connected (see e.g. \cite{LevPerWil09}). Then we use the standard detail balance condition to obtain the stationary distribution. 

\subsubsection{Connectivity of $\MCblocks$.} 

To prove that the chain $\MCblocks$ is connected (from every \emph{building sequence} $a$ we can reach every other building sequence $a'$ in a maximum of $\bigO(A)$ operations) we argue in two steps. 
Intuitively, we show that we can transform any two \emph{paths} into each other by some operations depicted in Figure~\ref{fig:path-moves}. 
Then it can be seen that every operation in Figure~\ref{fig:path-moves} corresponds to a sequence of operations in the Markov chain $\MCblocks$, given in Figure~\ref{fig:basic-moves}.
Formally:

Every path of width $n$ and area $A\leq \left \lfloor n^2/4 \right \rfloor $ can be turned into any other path of the same area and width by using the operations in Figure~\ref{fig:path-moves}. 
To prove this we consider a \emph{canonic path} for a given width $n$ and area $A$. 
The canonic path is the uniquely defined path $\mz \in \MZ(n,A)$ for which the following holds:
For every $i$, after $i$ steps (i.e. between $x=0$ and $x=i$) $\mz$ has maximum area among all paths in $\MZ(n,A)$.
The possible forms of the canonic path 
are shown in Figure~\ref{fig:connectivity-1}. 
Any given path with width $n$ and area $A$ can be transformed into the canonic path of the same area using the steps from Figure~\ref{fig:path-moves}. 
We overlay the given path with the canonic path and proceed in steps to the right as is schematically shown in Figure~\ref{fig:connectivity-1} with the black path being the given path and the red path being the canonic path.
There are three possibilities. Either the paths coincide, in which case we proceed to the right, or the given path differs proceeding with a $D$ move or with an $H$ move.
In both cases the given path must intersect the canonic path on the falling part because otherwise the area cannot be the same. Now we use the operations in
Figure~\ref{fig:path-moves} in horizontal sweeps from left
to right to fill-in the missing area of the canonic path. At
the end both paths must coincide because the areas are the
same.

Each of these operations can be simulated by some
operations on the sequences in
Figure~\ref{fig:basic-moves}. This can be seen immediately
because the four cases in Figure~\ref{fig:path-moves}
correspond directly to one or two operations in
Figure~\ref{fig:basic-moves}.

\begin{figure}[t!]
	\includegraphics[width=\textwidth]{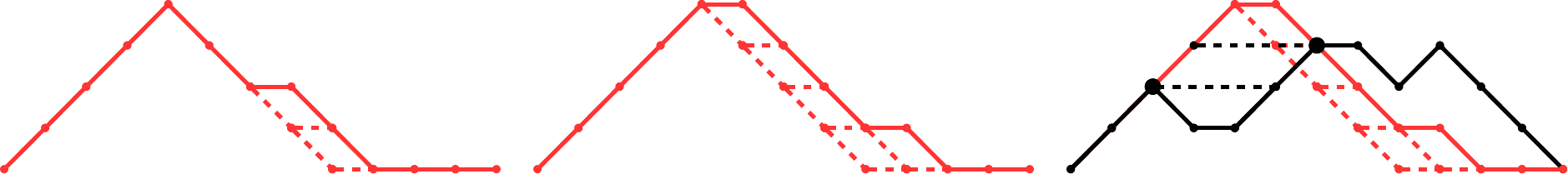}	
	\caption{ 
		(left, middle) The canonic path for $n=12$ and $A=18$, resp. $A=21$. 
		(right) Building a canonic path (red) from a given path (black) with same area.}
	\label{fig:connectivity-1}
\end{figure}

\begin{figure}
	\begin{subfigure}{0.49\textwidth}
		\centering
		\includegraphics[scale=.7]{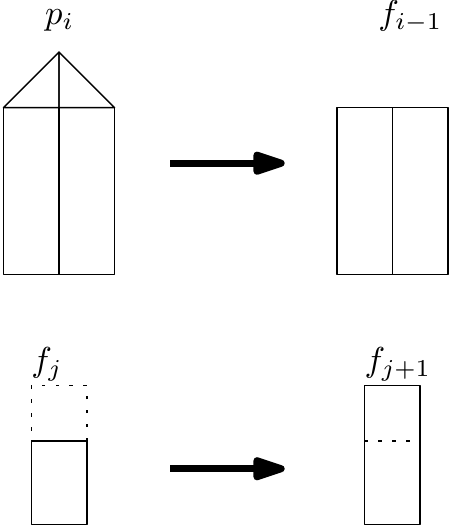}
		\caption{Peak to Flat.}
		\label{fig:PF}
	\end{subfigure}
	\begin{subfigure}{0.49\textwidth}
		\centering
		\includegraphics[scale=.7]{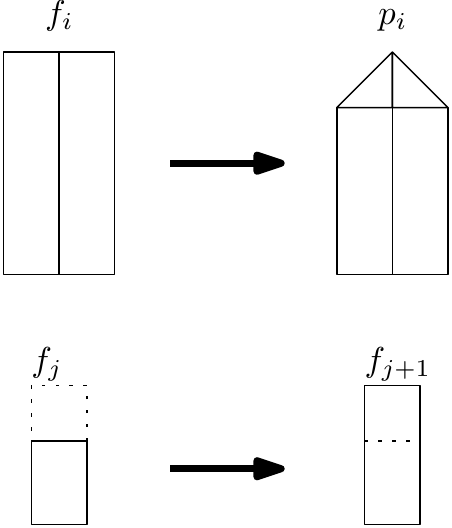}
		\caption{Flat to Valley.}
		\label{fig:FV}
	\end{subfigure}
	
	\begin{subfigure}{0.49\textwidth}
		\centering
		\includegraphics[scale=.7]{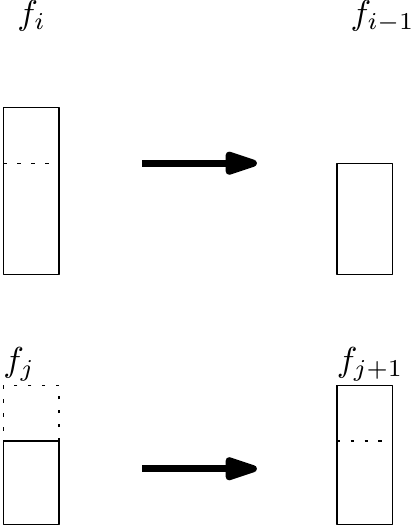}
		\caption{Flat to Flat.}
		\label{fig:FF}
	\end{subfigure}
	\begin{subfigure}{0.49\textwidth}
		\centering
		\includegraphics[scale=.7]{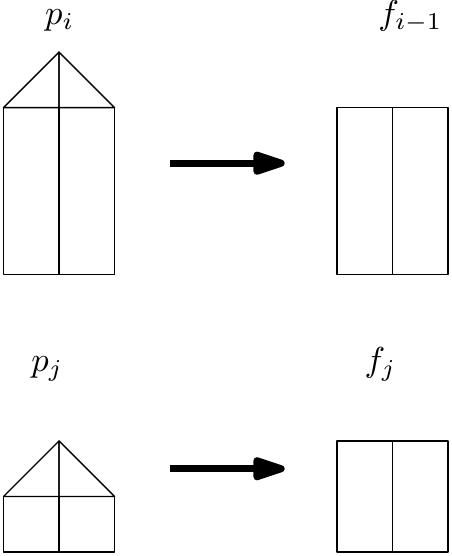}
		\caption{Peak into Valley.}
		\label{fig:PV}
	\end{subfigure}
	\caption{The basic operations over the sequences.}
	\label{fig:basic-moves}
\end{figure}

\begin{figure}
	\centering
	\includegraphics[scale=.8]{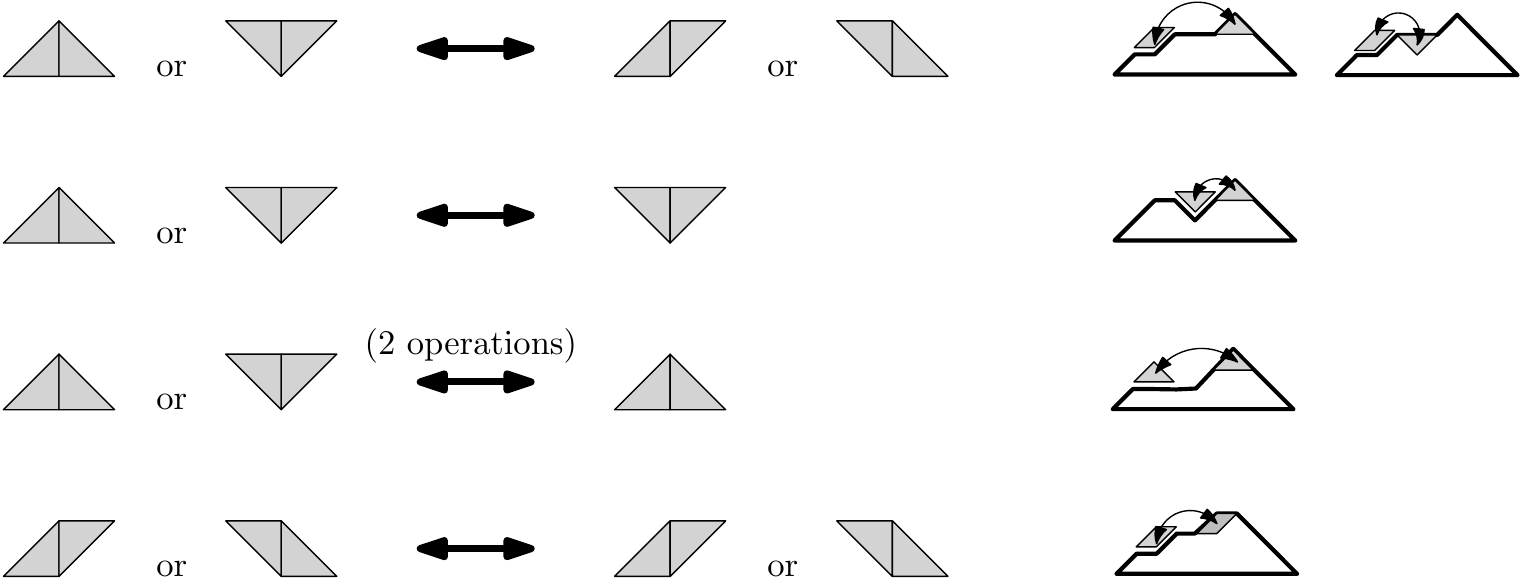}
	\caption{Changing \emph{the path}. The shape on one side is transformed to the shape on the other side. In some cases we need two of the operations defined in Figure~\ref{fig:basic-moves}.}
	\label{fig:path-moves}
\end{figure}

\subsubsection{Stationary Distribution of $\MCblocks$.}
It is well-known that the Metropolis chains with acceptance probability $A(a,a')=\min\left\{1,\frac{\pi(a')}{\pi(a)}\right\}$ have stationary distribution $\pi$ since the detailed balance condition is obviously satisfied: given that the number of operations is $k$, we have 
\begin{linenomath}
\begin{align*}
	P(a,a') =& \frac{A(a,a')}{2k} & \text{and}& & P(a',a) =& \frac{A(a',a)}{2k}
\end{align*}
\end{linenomath}
and the definition of $A(\cdot,\cdot)$ yields the detailed balance condition, that is,
$\pi(a) P(a,a')=\pi(a')P(a',a).$
\end{proof}

\subsubsection{Experimental Evaluation of $\MCblocks$}
We are interested in the required number of steps until the distribution of $\MCblocks$ is sufficiently close to its stationary distribution.
We measure the distance between two distributions by the \emph{total variation distance}.
The \emph{mixing time} of a Markov chain is the smallest time $t$ such that the total variation distance between the stationary distribution and the distribution after $t$ steps, starting from any state, is smaller than some small $\epsilon > 0$.

We study the mixing time of $\MCblocks$ for a given area $A$ and a given width $n$ by running the following experiment. 
We estimate the distribution after a given number of steps by repeatedly running $\MCblocks$ with an initial state $a_0$ defined as follows: 
The building block sequence consists of one peak of height $h$ for every $h\le \lfloor \sqrt{A} \rfloor$ and the remaining area and width is filled greedily with flats of maximal possible height.
The total variation distance of the distribution of $\MCblocks$ after some number of steps $t$ from its stationary distribution $\pi$ is 
\[ d_{TV}(P^t(a_0,\cdot),\pi) = \frac{1}{2} \cdot \sum_{a\in S(n,A)} \Big|P^t(a_0,a)-\pi(a)\Big|. \]
We estimate the mixing time for a given area $A$ and a given width $w$ by computing the total variation distance for increasing $t$ until the total variation distance is below $0.05$.
\begin{figure}[t]
	\centering
	\includegraphics[width=.49\linewidth]{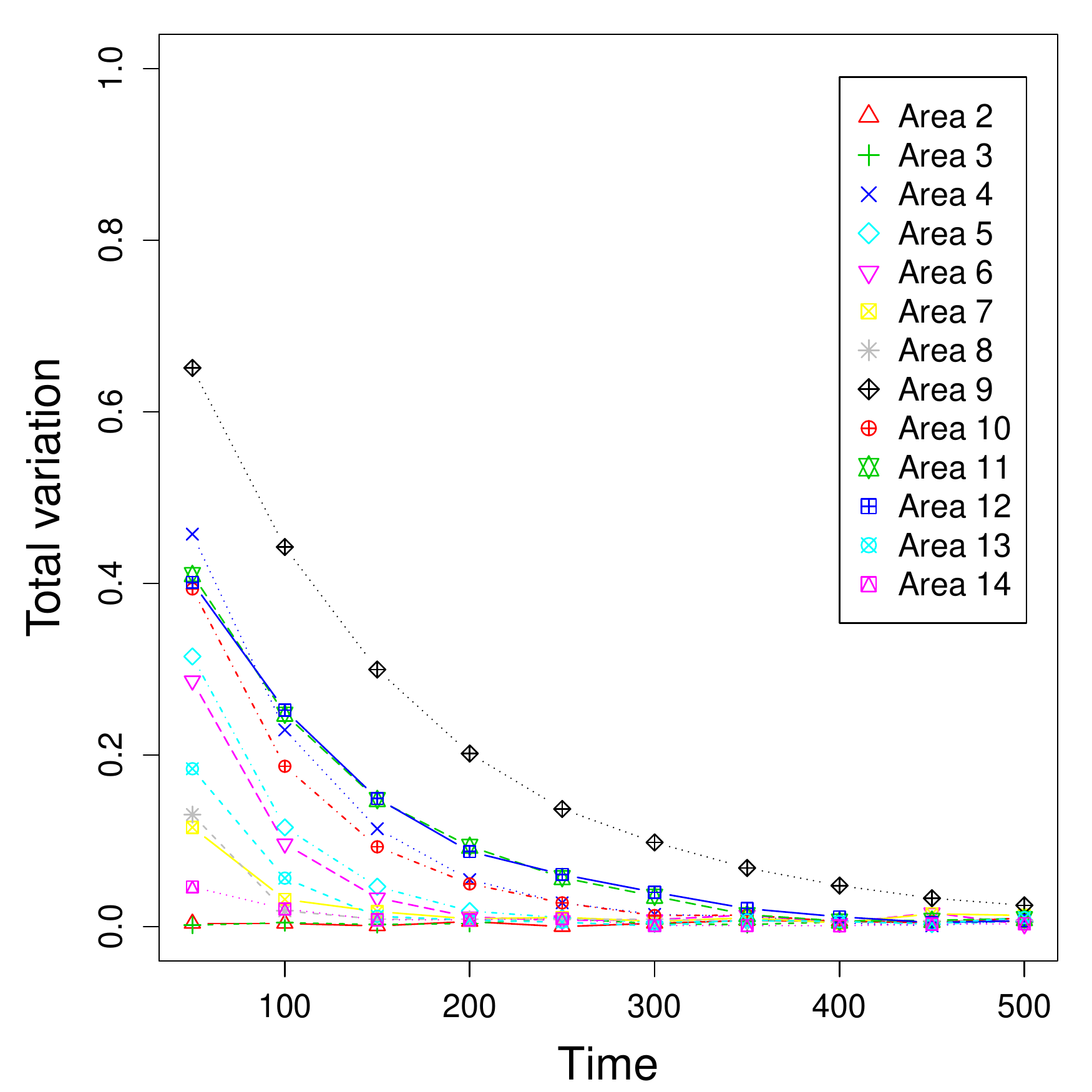}
	\hfill%
	\includegraphics[width=.49\linewidth]{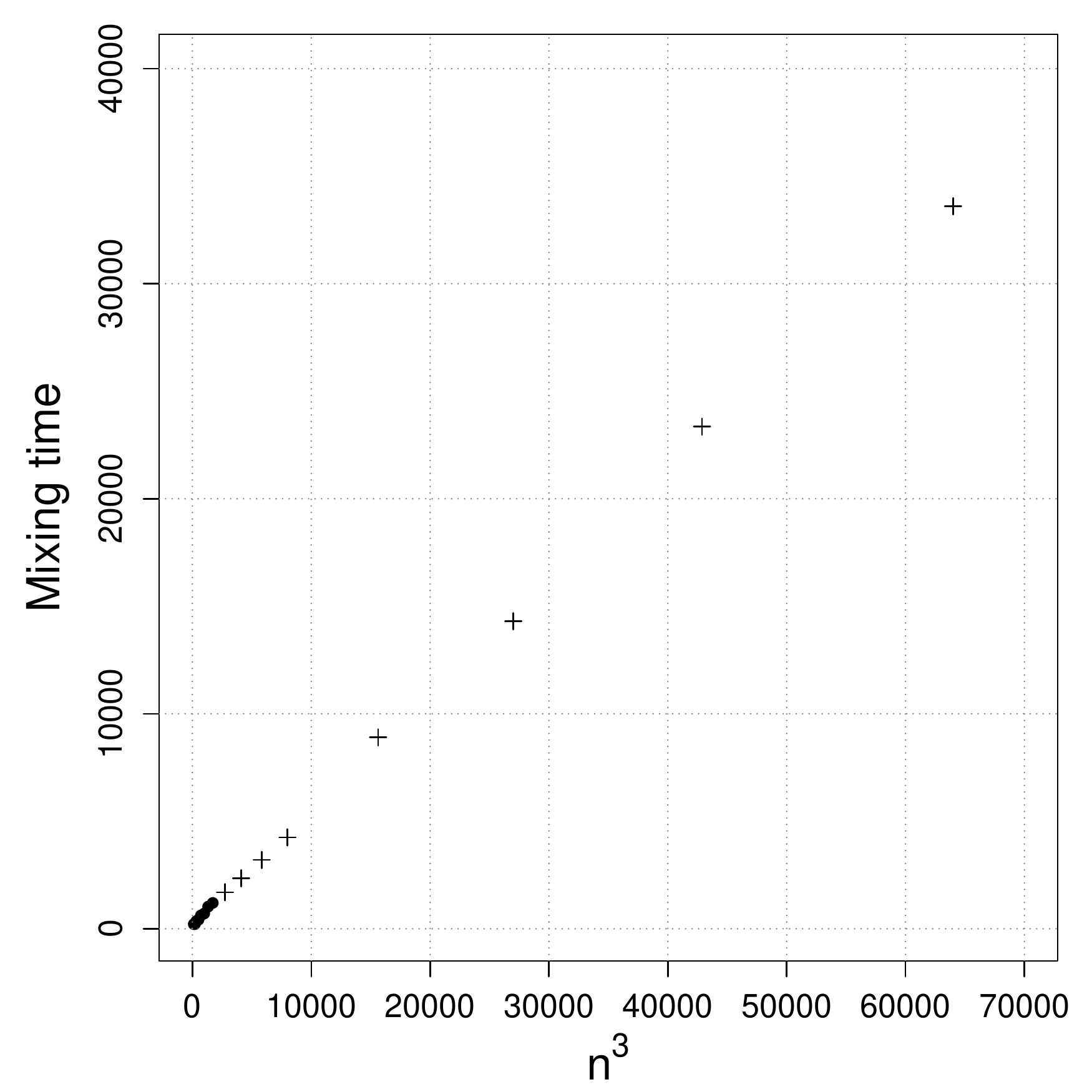}
	\caption{(left) Total variation for $n=8$ and all $A \le (n/2)^2$ with $|S(8,A)|>1$. 
		(right) Maximal mixing time for given widths ($\bullet$), mixing time for areas $A^*_n$ ({\small+}).}
	\label{fig:experiments}
\end{figure}

\renewcommand{\arraystretch}{1.3}
\setlength{\tabcolsep}{0.5em}
\begin{table}
	\centering
	\begin{tabular}{llll}
		\toprule
		\emph{width} & \emph{area}        & \emph{runs} & \emph{mixing time} \\ \midrule
		4 to 12      & all relevant areas & 10\,000        & 200 to 1\,200        \\
		14           & 36                 & 100\,000       & 1\,700             \\
		16           & 49                 & 100\,000       & 2\,350             \\
		18           & 64                 & 100\,000       & 3\,200             \\
		20           & 81                 & 100\,000       & 4\,250             \\
		25           & 144                & 100\,000       & 8\,900             \\
		30           & 196                & 100\,000       & 14\,300            \\
		35           & 289                & 400\,000       & 23\,350            \\
		40           & 361                & 500\,000       & 33\,600            \\ \bottomrule\hline
	\end{tabular}
	\vspace{1em}
	\caption{Experimental setup: The distribution after $t$ steps was estimated using 10\,000 to 500\,000 runs of $\MCblocks$. The maximal mixing time was estimated by computing the mixing time for $A^*_n$ for widths larger than 12.}
	\label{tab:experiments}
\end{table}

Figure~\ref{fig:experiments} (left) illustrates the mixing time for width 8 and every area $A$ with more than one possible building block sequence.
The maximal mixing time (400 steps) is necessary for area $9$. 
In fact, for every width smaller than 13, the mixing time is maximal for area $A^*_n = ((n-2)/2)^2$ if $n$ is even and $A^*_n = ((n-1)/2)^2$ otherwise. 
This is due to our choice of the initial state of $\MCblocks$.
We estimate the maximal mixing time for widths larger than 12 by computing the mixing time for $A^*_n$ \emph{only}, as the number of repeats necessary to estimate the distribution of $\MCblocks$ after $t$ steps depends on the number of possible building block sequences, which grows exponentially depending on $n$.
Figure~\ref{fig:experiments} (right) shows the maximal mixing time up to width 40. 
The plot suggests that the number of steps necessary to approximate the stationary distribution does not grow exponentially depending on the width $n$, the algorithm is probably faster than the sampler based on dynamic programming and the results suggest that the MCMC sampler achieves the mixing time $\bigO(n^3)$.

\begin{conjecture}
\label{conj:cubic-mixing-time}
$\MCblocks$ mixes in time $\bigO(n^3)$.
\end{conjecture}

\begin{remark}
The implementation of $\MCblocks$ is available at \REPOURL.
\end{remark}

% ---- Bibliography ----
%
\bibliographystyle{plain}
\bibliography{bibliography}

\begin{thebibliography}{10}

\bibitem{barcucci1995construction}
Elena Barcucci, Alberto Del~Lungo, Elisa Pergola, and Renzo Pinzani.
\newblock A construction for enumerating k-coloured motzkin paths.
\newblock In {\em Computing and Combinatorics}, pages 254--263. Springer, 1995.

\bibitem{Bub11}
Russ Bubley.
\newblock {\em Randomized Algorithms: Approximation, Generation and Counting}.
\newblock Springer, 2011.

\bibitem{oesis_motzkin_area}
Emeric Deutsch and Alois~P Heinz.
\newblock A129181 motzkin paths by area, {Online Encyclopedia of Integer
  Sequences}.
\newblock \url{http://oeis.org/A129181}, June 2012.

\bibitem{Deza98metricsonpermutations}
Michael Deza and Tayuan Huang.
\newblock Metrics on permutations, a survey.
\newblock {\em Journal of Combinatorics, Information and System Sciences},
  1998.

\bibitem{diaconis1977spearman}
Persi Diaconis and Ronald~L Graham.
\newblock Spearman's footrule as a measure of disarray.
\newblock {\em Journal of the Royal Statistical Society. Series B
  (Methodological)}, pages 262--268, 1977.

\bibitem{donaghey1977motzkin}
Robert Donaghey and Louis~W Shapiro.
\newblock Motzkin numbers.
\newblock {\em Journal of Combinatorial Theory, Series A}, 23(3):291--301,
  1977.

\bibitem{oesis_triangle}
Olivier G{\'e}rard, Mathieu Guay-Paquet, and Alois~P Heinz.
\newblock A062869 permutation with fixed total displacement, {Online
  Encyclopedia of Integer Sequences}.
\newblock \url{https://oeis.org/A062869}, May 2014.

\bibitem{goulden2004combinatorial}
Ian~P Goulden and David~M Jackson.
\newblock {\em Combinatorial enumeration}.
\newblock Dover Publications, 2004.

\bibitem{GrePasRan09}
Sam Greenberg, Amanda Pascoe, and Dana Randall.
\newblock Sampling biased lattice configurations using exponential metrics.
\newblock In {\em 20th {ACM-SIAM} Symposium on Discrete Algorithms SODA'09},
  pages 76--85, 2009.

\bibitem{guay2014generating}
Mathieu Guay-Paquet and Kyle Petersen.
\newblock The generating function for total displacement.
\newblock {\em The Electronic Journal of Combinatorics}, 21(3):P3--37, 2014.

\bibitem{humphreys2010history}
Katherine Humphreys.
\newblock A history and a survey of lattice path enumeration.
\newblock {\em Journal of statistical planning and inference},
  140(8):2237--2254, 2010.

\bibitem{samplingpermutations}
Ekhine Irurozki.
\newblock {\em Sampling and learning distance-based probability models for
  permutation spaces}.
\newblock PhD thesis, University of the Basque Country, Donostia - San
  Sebasti\'{a}n, July 2014.

\bibitem{donald1999art}
Donald~E. Knuth.
\newblock The art of computer programming.
\newblock {\em Sorting and searching}, 3:426--458, 1999.

\bibitem{LevPerWil09}
David~Asher Levin, Yuval Peres, and Elizabeth~Lee Wilmer.
\newblock {\em {M}arkov chains and mixing times}.
\newblock American Mathematical Soc., 2009.

\bibitem{merlini2003generating}
Donatella Merlini.
\newblock Generating functions for the area below some lattice paths.
\newblock In {\em Discrete Random Walks, DRW'03}, pages 217--228, 2003.

\bibitem{pergola2002bijective}
E~Pergola, R~Pinzani, S~Rinaldi, and RA~Sulanke.
\newblock A bijective approach to the area of generalized motzkin paths.
\newblock {\em Advances in Applied Mathematics}, 28(3):580--591, 2002.

\bibitem{sulanke2000moments}
Robert~A Sulanke.
\newblock Moments of generalized motzkin paths.
\newblock {\em Journal of Integer Sequences}, 3(00.1):1--14, 2000.

\end{thebibliography}

\newpage
\section*{Appendix}
\appendix

\section{Postponed Proofs}
\label{app:omitted-proofs}

\setcounter{corollary}{0}
\pathtopermcorollary

\begin{proof}
Given a Motzkin path $mz$, we can sample a permutation u.a.r.\ among all
permutations that map into $mz$. For this, consider $mz$ as a (feasible) sequence of letters $U$, $D$ and $H$ (denoting diagonally \emph{Up-right} moves,
diagonally \emph{Down-right} moves and \emph{Horizontal-right} moves). Then do the following:
\begin{enumerate}
	\item	Scan the sequence from left to right. When a new $D$ is found, match it with any of the $U$ on the left that are not yet matched to
		any $D$ (choose such a $U$ u.a.r.). This step constructs left-to-right edges from $U$ to $D$:
		\begin{center}
			\includegraphics{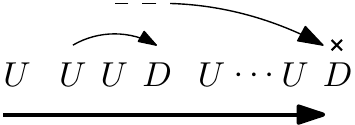}
		\end{center}
	\item	Scan from right to left the sequence, matching a newly encountered $U$ with any $D$ on the right not yet matched in this step (the
		matchings of the previous ``left-to-right'' step do not count). This step is the symmetric of the previous one and it constructs
		right-to-left edges from $D$ to $U$.
	\item	For any $H$ we choose a ``fixed point'', a ``left-to-right'' or a ``right-to-left'' edge, meaning the following: The number of
		left-to-right edges crossing this $H$ -- corresponding to a flat of height $i$ -- is equal to the number $i$ of right-to-left edges also
		crossing this $H$ (this property is due to the ``balanced'' matchings of $U$ and $D$). So there are $2i+1$ options for $H$, where the
		$2i$ options correspond to breaking one of the $i$ left-to-right edges, or one of the $i$ right-to-left edges. The last option is to let the $H$ map to a trivial cycle of the permutation. We choose one of these $2i+1$
		options u.a.r.
\end{enumerate}
\end{proof}

\section{Experiments}
We estimate the distribution after a given number of steps by repeatedly running $\MCblocks$ for an appropriate number of steps depending on the width (cf. Table~\ref{tab:experiments}).
To compute the total variation distance of the distribution of $\MCblocks$ after $t$ steps from its stationary distribution $\pi$, i.e.
\[ d_{TV}(P^t(a_0,\cdot),\pi) = \frac{1}{2} \cdot \sum_{a\in S(n,A)} \Big|P^t(a_0,a)-\pi(a)\Big|, \]
we do not need to know all building block sequences in $S(n,A)$. 
Let $A_t\subseteq S(n,A)$ be the set of all building block sequences reached by $\MCblocks$ after $t$ steps in at least one run.
Then, the total variation distance after $t$ steps is 
\[ d_{TV}(P^t(a_0,\cdot),\pi) = \frac{1}{2} \cdot \Big(\, \big(D(n,A) -  \sum_{a\in A_t} (m(a) \cdot perm(a))\big) + \sum_{a\in A_t} \Big|P^t(a_0,a)-\pi(a)\Big|\, \Big). \]

\begin{figure}
	\centering
	\includegraphics[width=.49\linewidth]{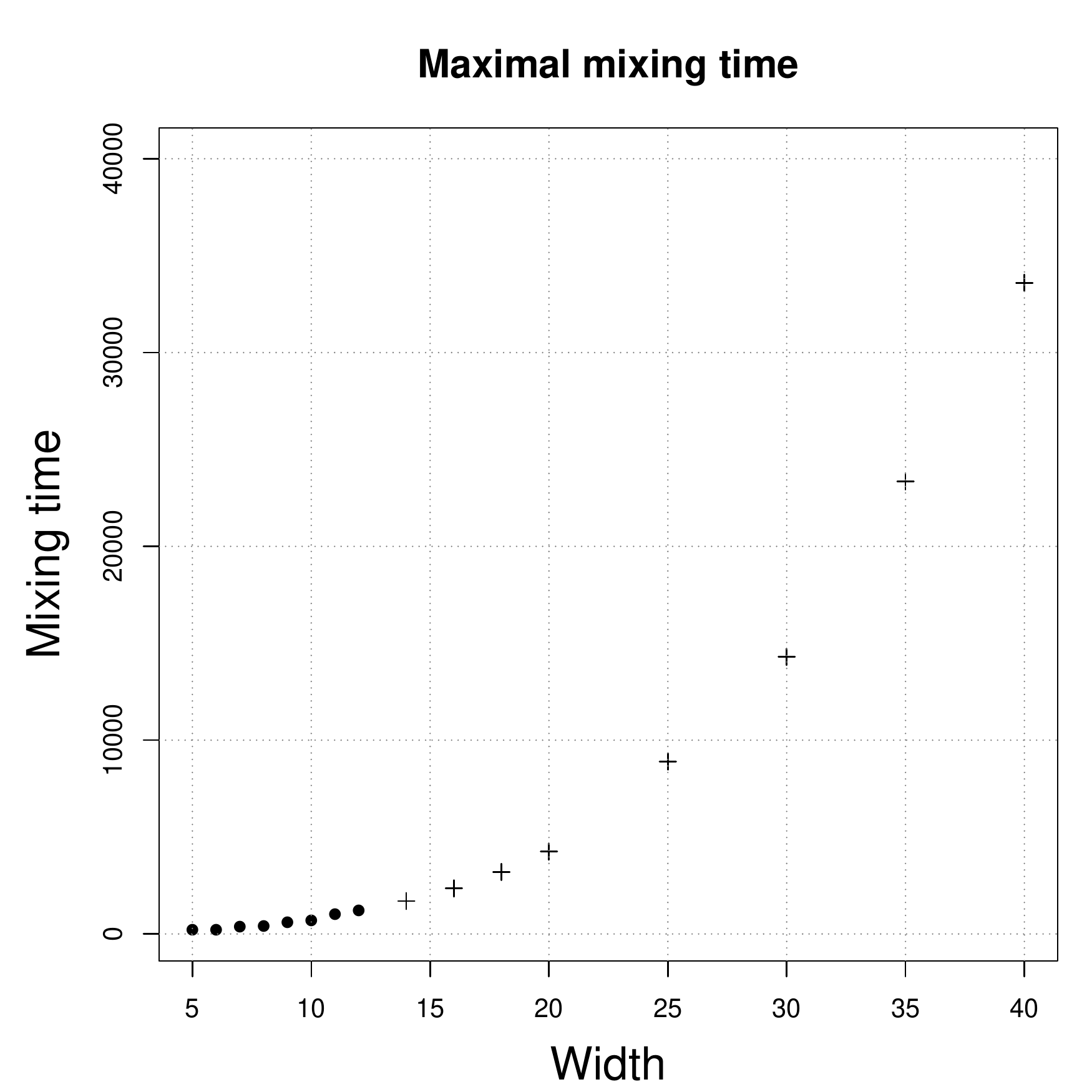}
	\hfill%
	\includegraphics[width=.49\linewidth]{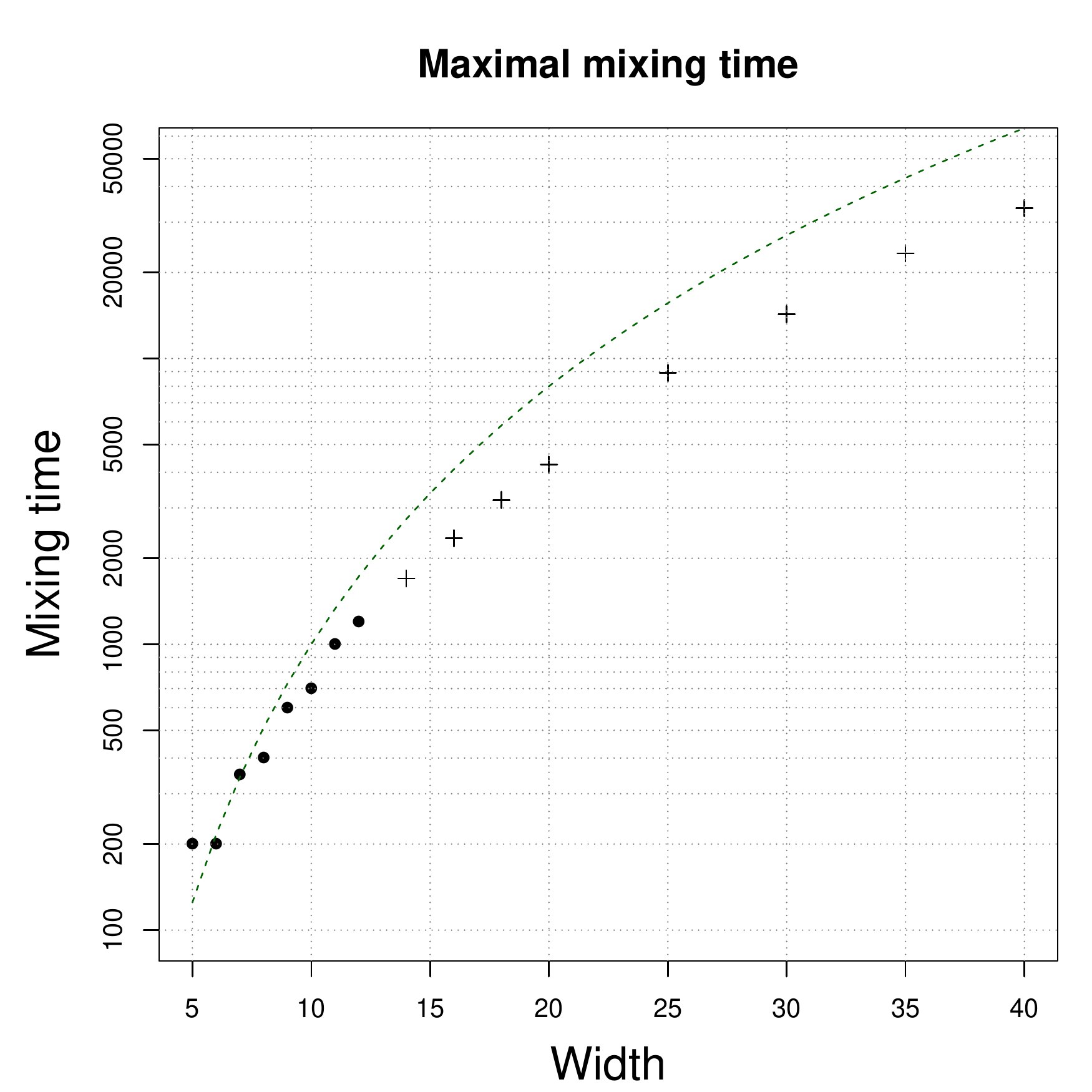}
	\caption{Plot of the maximal mixing time for a given width (filled circle) and mixing times for area $A^*_n$ (cross) with a linear y-axis (left) and a logarithmic y-axis (right). The dashed line in the right plot shows $f(x) = x^3$.}
\end{figure}

\begin{figure}
	\centering
	\includegraphics[width=.49\linewidth]{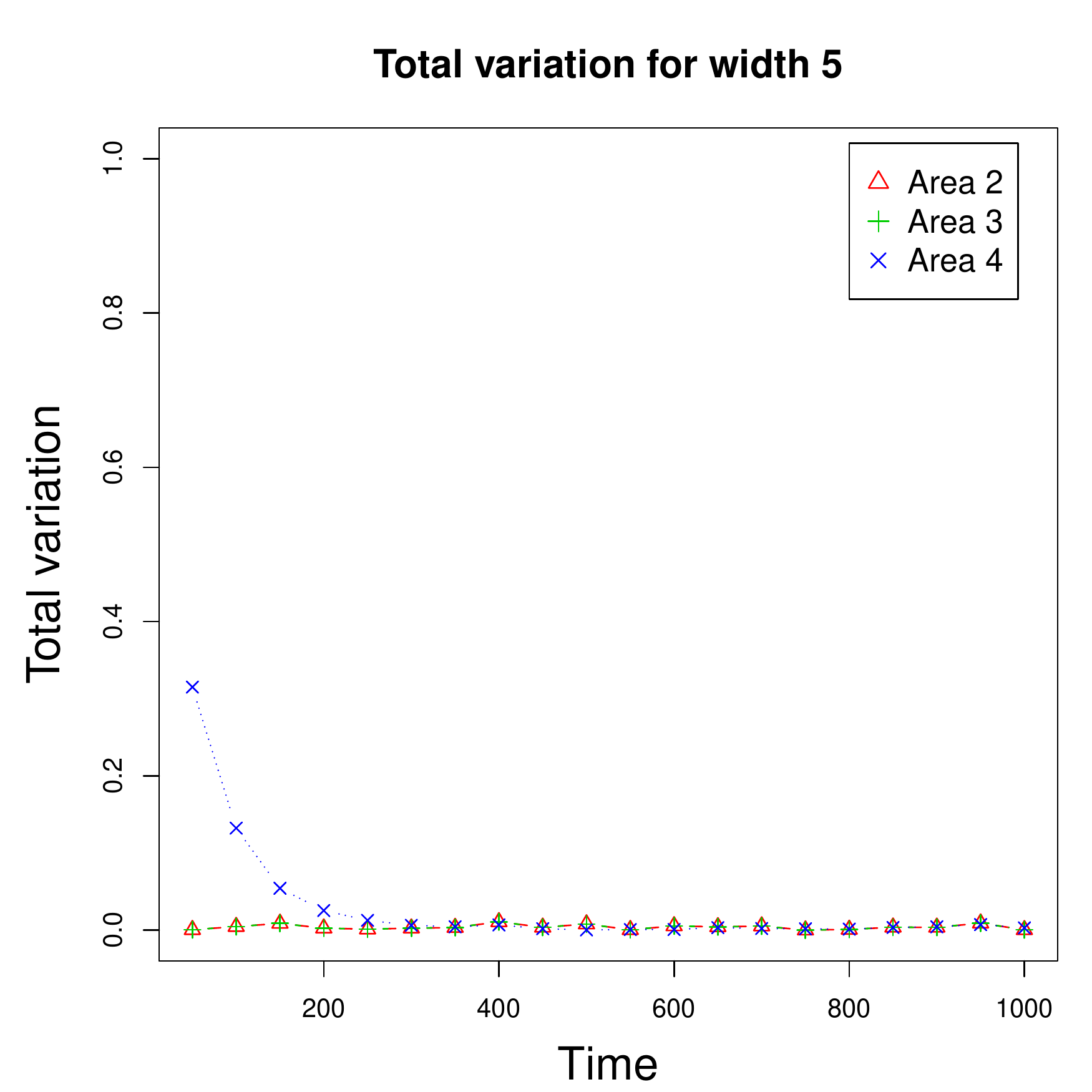}
	\hfill%
	\includegraphics[width=.49\linewidth]{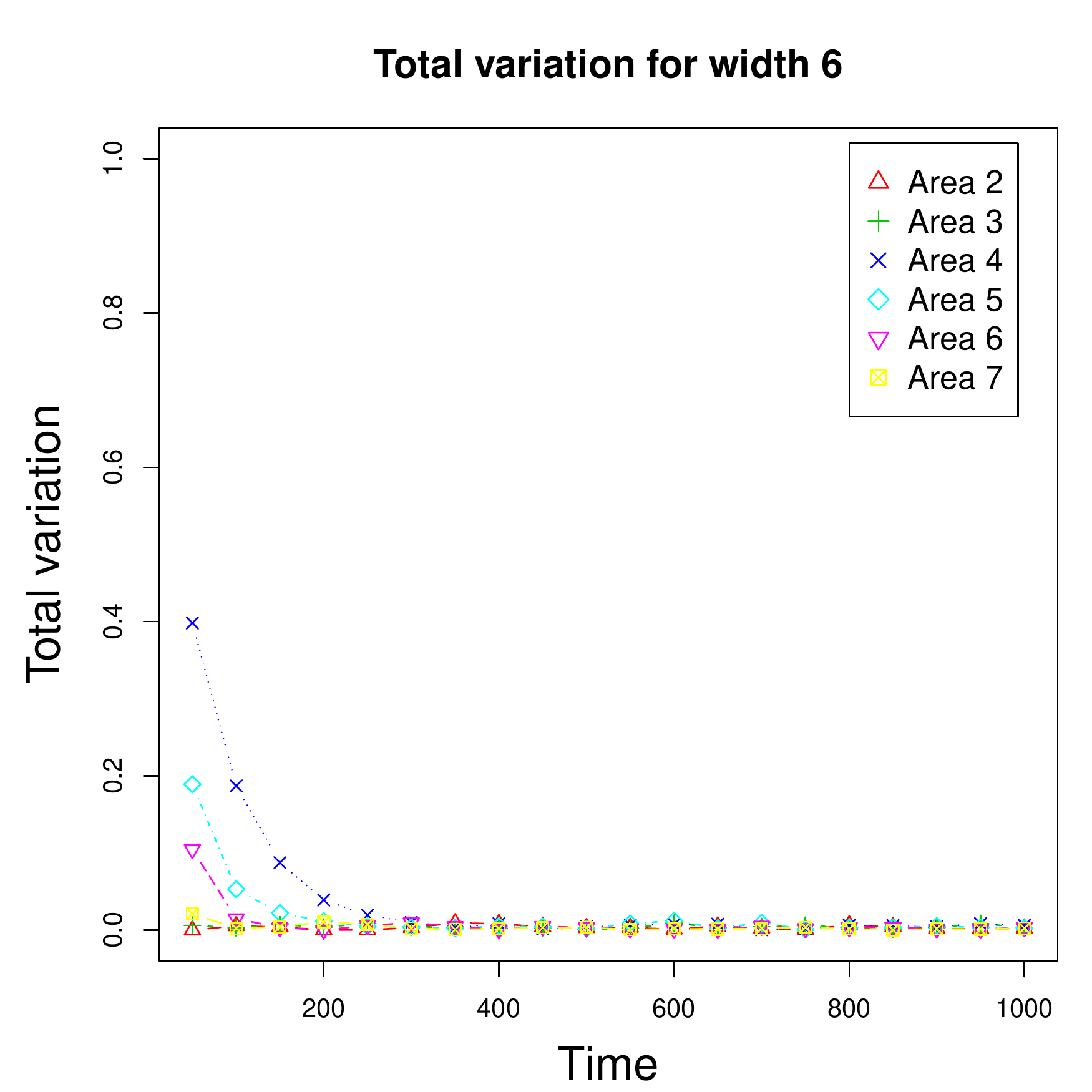}
	\caption{Total variation for width $n=5$ (left), respectively $n=6$ (right), and every area $A \le (n/2)^2$ with $|S(n,A)|>1$.}
\end{figure}

\begin{figure}
	\centering
	\includegraphics[width=.49\linewidth]{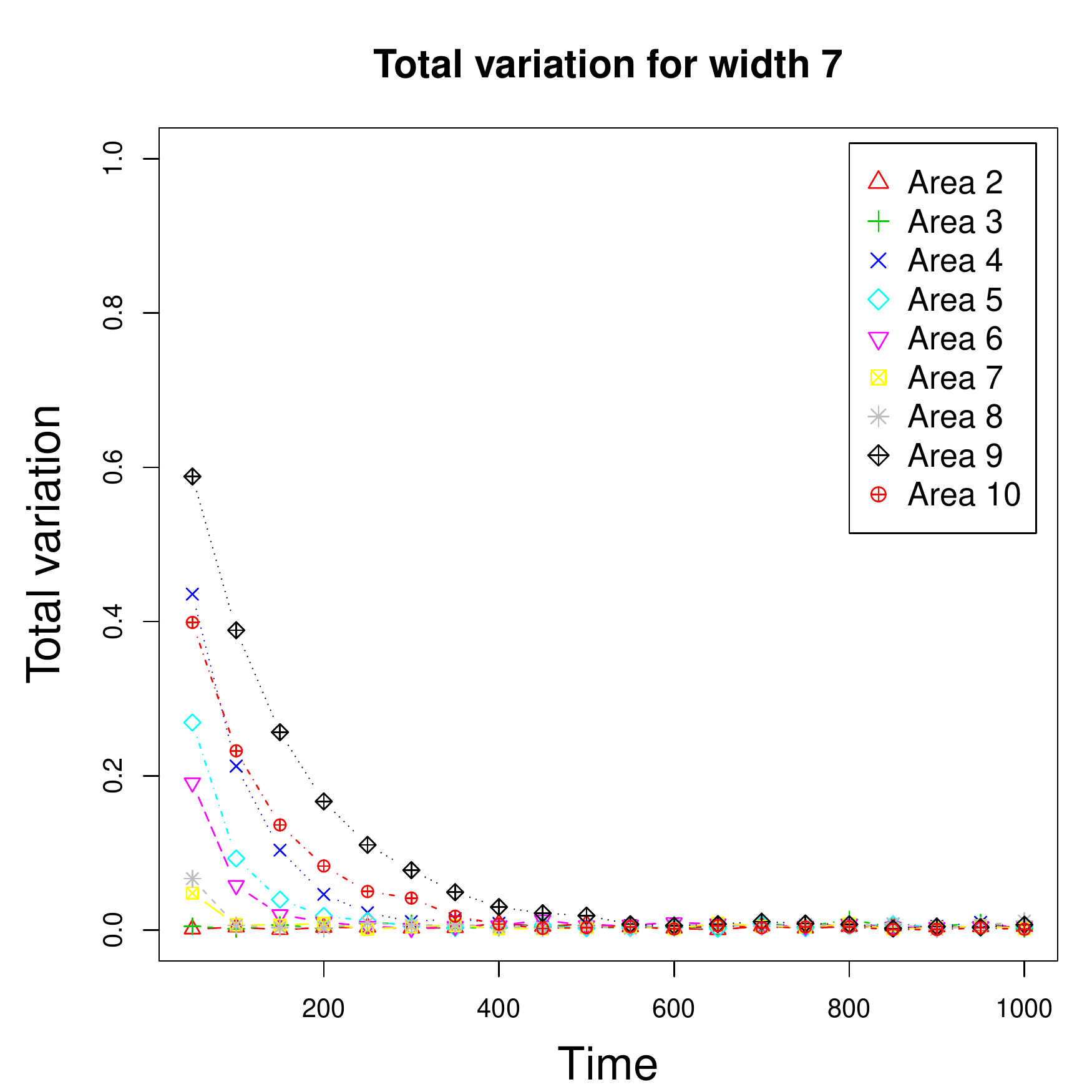}
	\hfill%
	\includegraphics[width=.49\linewidth]{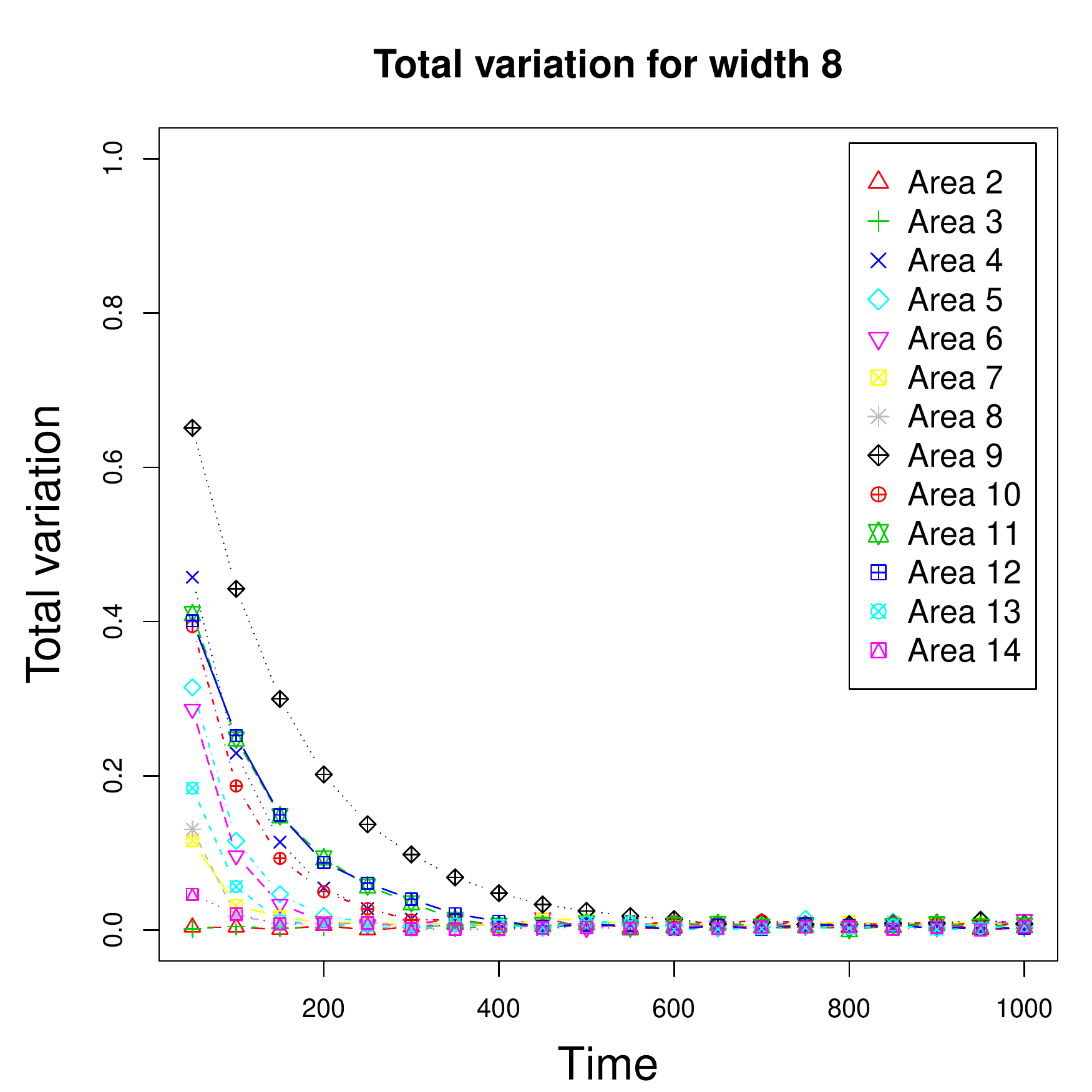}
	\caption{Total variation for width $n=7$ (left), respectively $n=8$ (right), and every area $A \le (n/2)^2$ with $|S(n,A)|>1$.}
\end{figure}

\begin{figure}
	\centering
	\includegraphics[width=.49\linewidth]{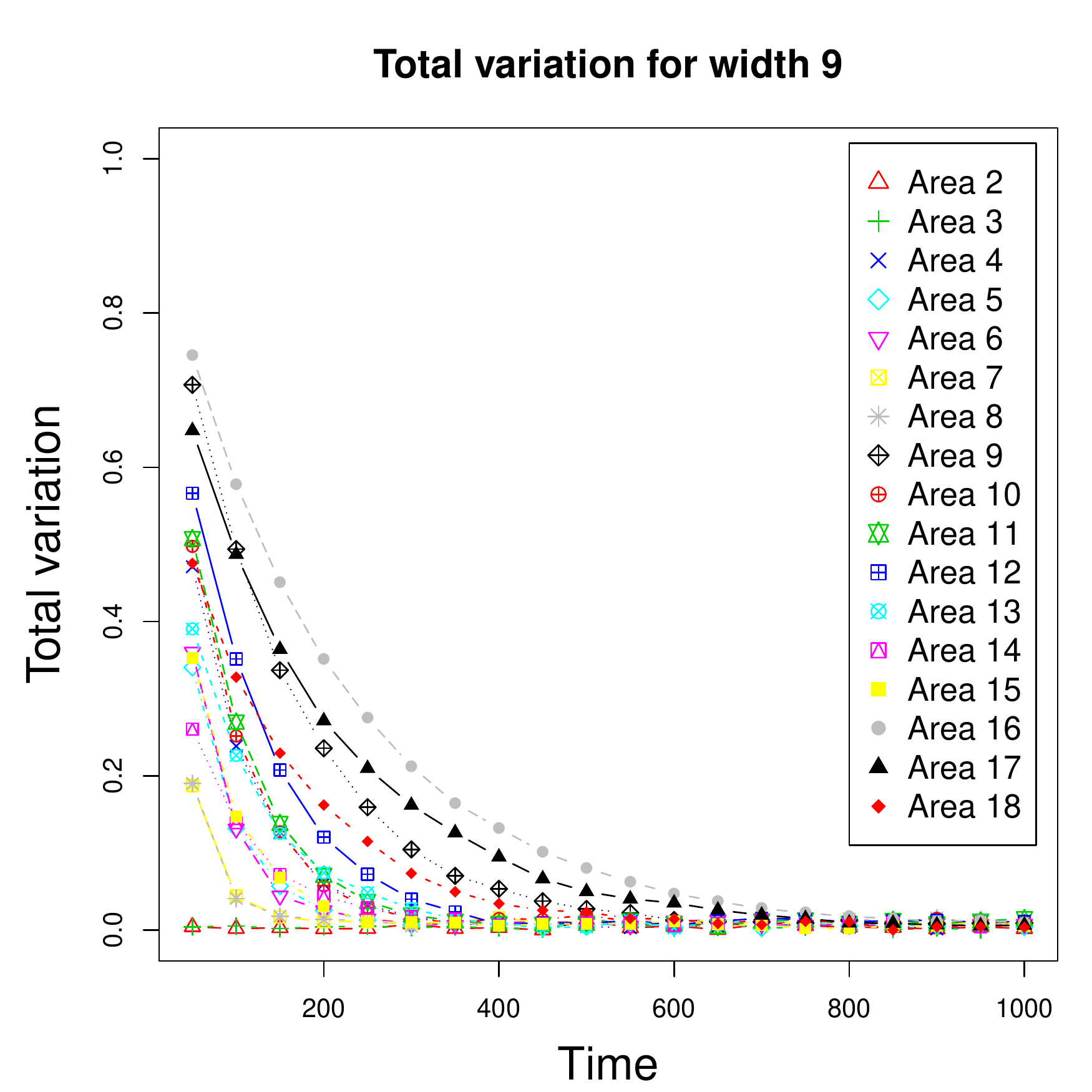}
	\hfill%
	\includegraphics[width=.49\linewidth]{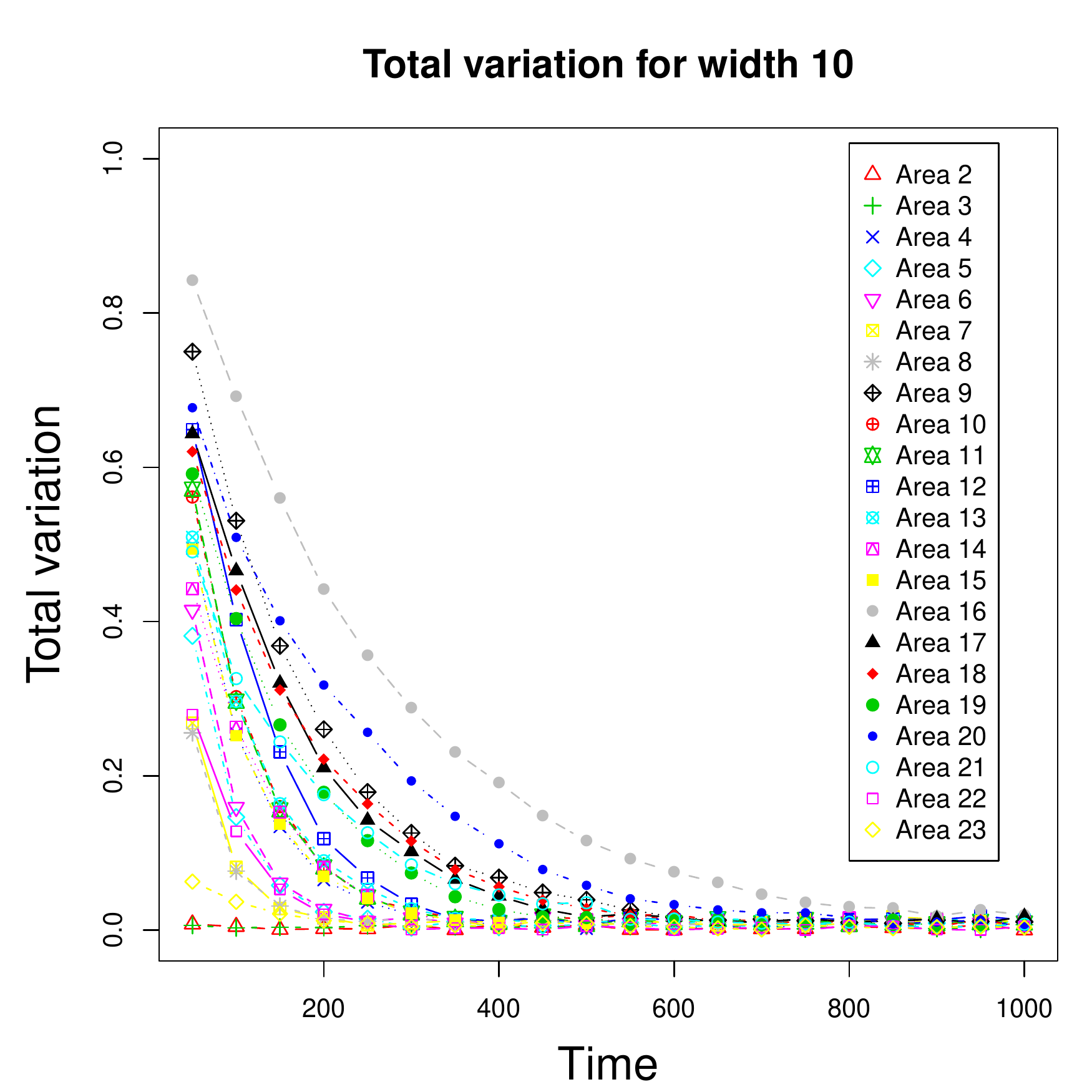}
	\caption{Total variation for width $n=9$ (left), respectively $n=10$ (right), and every area $A \le (n/2)^2$ with $|S(n,A)|>1$.}
\end{figure}

\begin{figure}
	\centering
	\includegraphics[width=.49\linewidth]{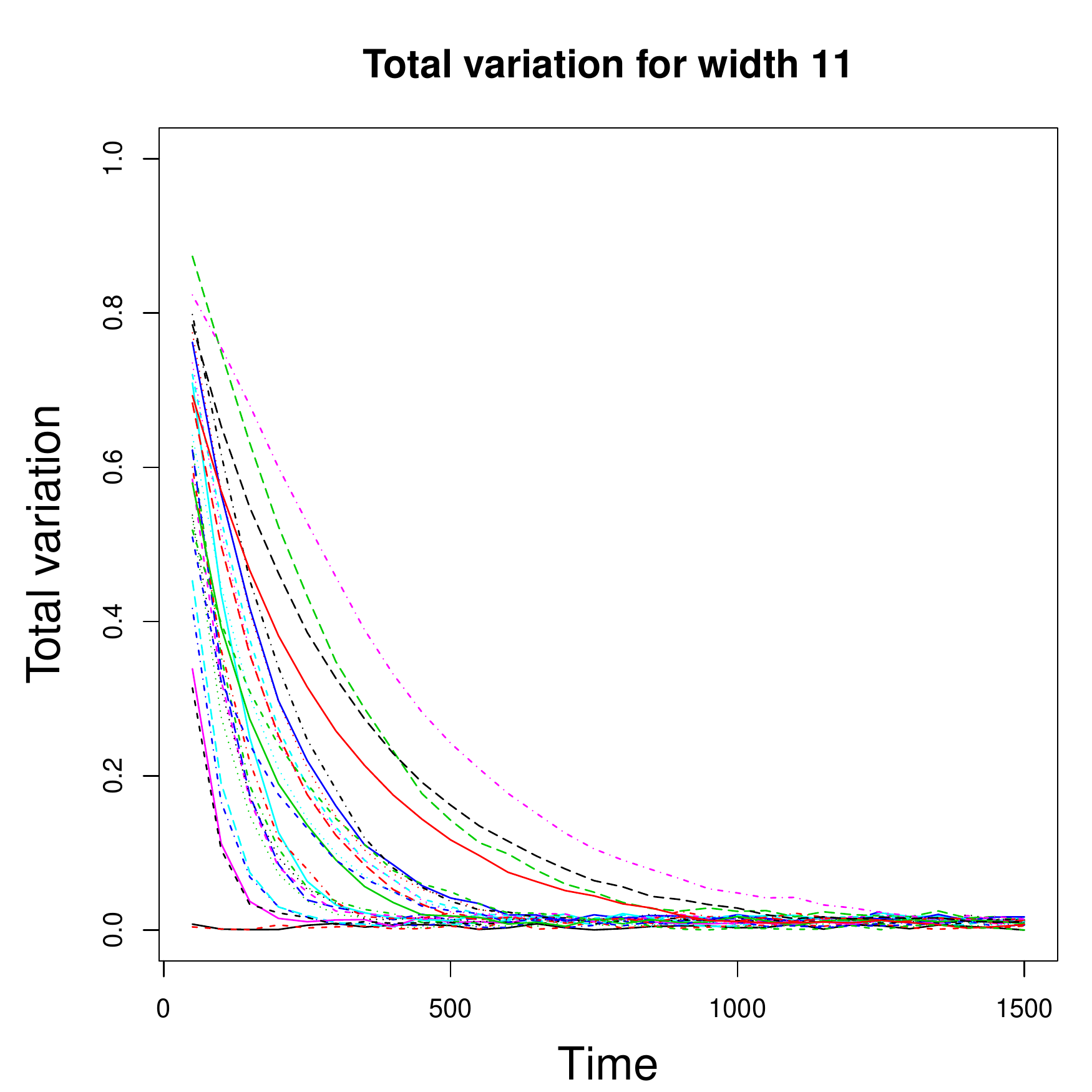}
	\hfill%
	\includegraphics[width=.49\linewidth]{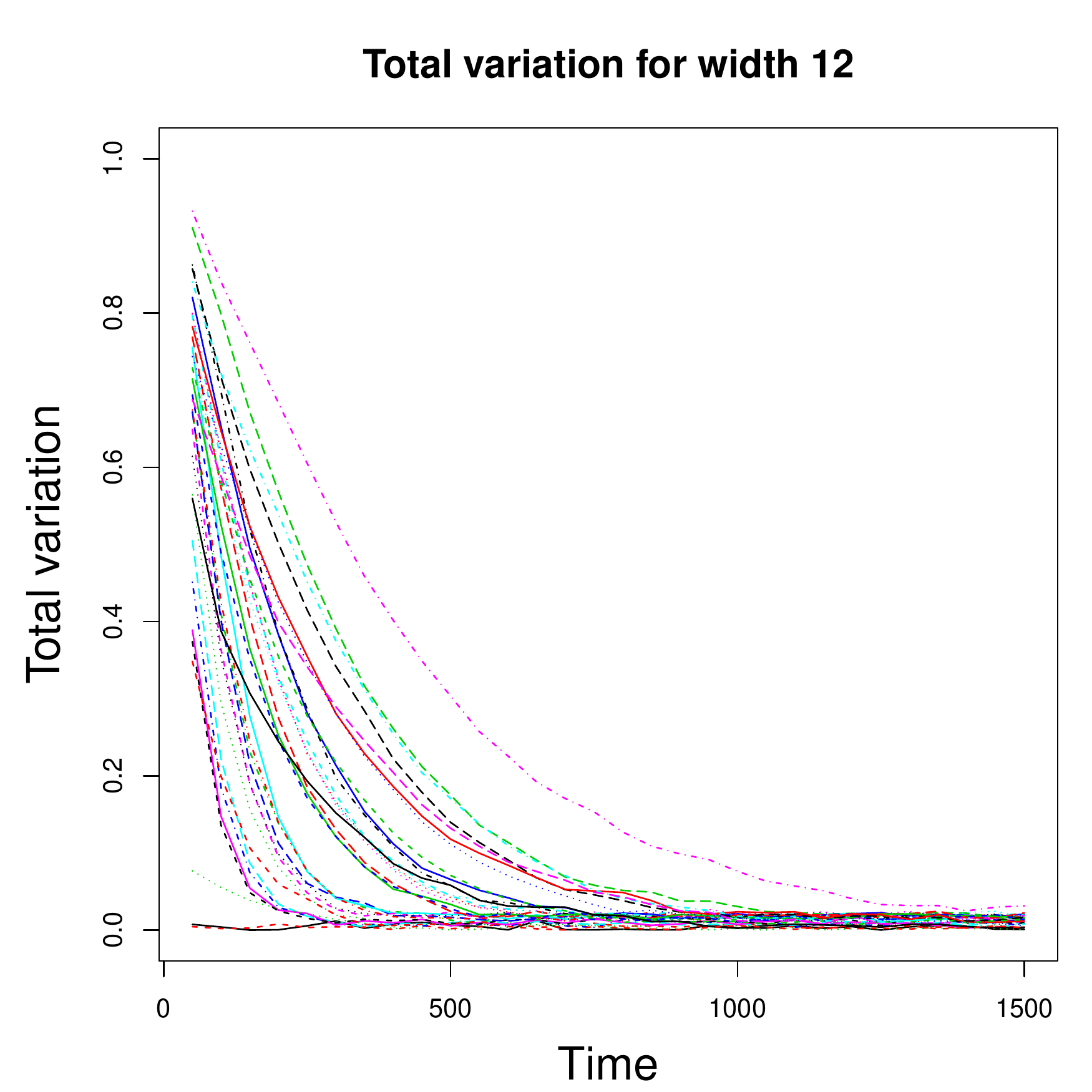}
	\caption{Total variation for width $n=11$ (left), respectively $n=12$ (right), and every area $A \le (n/2)^2$ with $|S(n,A)|>1$.}
\end{figure}

\newpage
\section{Top-Down Dynamic Programming Approach}
\label{app:DP-top-down}
For the sake of clarity we first describe the dynamic programming procedure for counting \emph{unweighted} Motzkin paths of any given width $n$ and area $A$.\footnote{Though this task can be solved by an easier dynamic programming, our approach does extend to the weighted case, which is our main goal.}
 Specifically, we define the following subproblem:
\begin{linenomath}
\begin{align*}
	M(n,A,h,p) =& \text{ the number of Motzkin paths of width $n$,}\\
	 & \text{area $A$, height $h$ and with $p_h = p$ and $f_h = 0$.}
\end{align*}
\end{linenomath}
Note that we forbid flats of height $h$ and fix the number of peaks at height $h$. To count the number of paths of a given $n$, $A$ and height $h$ including flats at height $h$ (so dropping the $f_h=0$ condition) and not fixing the number of peaks at height $h$ we can simply compute 
\[M(n,A,h) := M(n,A,h+1,0).\]
The number of Motzkin paths of given width and area can be obtained by summing over all $h$:
\[ M(n,A) := \sum_{h=0}^n M(n,A,h). \]
If we also sum over all possible areas, we get the classic Motzkin numbers
\[ M(n) := \sum_{A=0}^{n^2} M(n,A). \]
We next show that $M$ can be computed in polynomial time. We make use of Equation~\eqref{eq:ma-product} in Theorem~\ref{th:blocks_rearrange} for the number $m(a)$ of Motzkin paths that can be constructed out of building sequence $a$.
\begin{theorem}
	Table $M(n,A,h,p)$ can be computed for all the  $\mathcal{O}(n^5)$ many possible parameters in $\mathcal{O}(n^7)$ total time.
\end{theorem}

\begin{proof}
	Recall that in $M$ we forbid flats of height $h$ and fix the number of peaks of height $h$. This allows the following recursive counting:
	
	\begin{description}
		\item[Base case.] If any of the parameters is equal $0$ then
		\begin{linenomath}
		\begin{align}
			\label{eq:DP_base}
				& M(n,A,h,p) := & 
				\begin{cases}
					1 \quad \text{ if } n=A=h=p=0,\\
					0 \quad\text{ else if } n \leq 0 \text{ or } A \leq 0 \text{ or } h \leq 0.
				\end{cases}
		\end{align}
		\end{linenomath}
		\item[Recursion.] Otherwise, when all parameters are strictly positive, we have
	\end{description}
		\begin{align}	
		 	M(n,A,h,p) :=& 
			\sum_{f=0}^n {p+f\choose f} \cdot \left( \sum_{p' = 1}^{n/2} {p+f+p'-1 \choose p'-1} \cdot M(n',A',h-1,p')\right),  	\label{eq:DP}
	    \end{align}
	 where $n'=n-2p-f$ and $A' = A - (2h-1)p - (h-1)f$.
	
	In the recursive case, we enumerate all potential numbers of flats and peaks at height $h-1$ so that we can look up the corresponding subproblems.
	These subproblems are then weighted by the number of possible interleavings ${p+f\choose f}{p+f+p'-1 \choose p'-1}$ which we derived in Equations~\eqref{eq:peak_combos} and~\eqref{eq:flat_combos}.
	We have thus shown that the table $M(n,A,h,p)$ above can be computed for all the $\mathcal{O}(n^5)$ many possible parameter values in time $\mathcal{O}(n^7)$ 
	with the bottleneck being the two nested sums in Equation~\eqref{eq:DP}.
\end{proof}

Let us now consider the problem of counting \emph{weighted} Motzkin paths, that is, the function $D(n,d)$  in Equation~\eqref{eq:displacement_formula_via_paths}.
To this end, we extend the definition of $M$ above so to count each path $\mz$ according to its weight $perm(a^{(\mz)})$. 
The resulting table $D(n,d,h,p)$ can be computed recursively in a top-down fashion by incorporating into the recursion of $M(n,A,h,p)$ the two terms defining $perm(a^{(\mz)})$ in Equation~\eqref{eq:perm-product}:
\begin{linenomath}
\begin{align}
\begin{split}
D(n,d,h,p) :=& \underbrace{h^{2p}}_{\text{peak options}} \cdot \sum_{f=0}^n \underbrace{(2h-1)^{f}}_{\text{flat options}} \cdot {p+f\choose f} \cdot \\
&\left( \sum_{p' = 1}^{n/2} {p+f+p'-1 \choose p'-1} \cdot D(n',d',h-1,p')\right),
\end{split}
\end{align}
\end{linenomath}
where $n'=n-2p-f$ and $d' = d - (2h-1)p - (h-1)f$.
The base case is identical to the unweighted case Equation~\eqref{eq:DP_base}.
Again, we can drop the condition that there are no flats on the last level $h$:
\begin{linenomath}
\begin{align*}
	D(n,d,h) &:= D(n,d,h+1,0),  
	\intertext{and count all weighted Motzkin paths of given area and width simply as}  
	D(n,d) &:= \sum_{h=0}^n D(n,d,h). 
\end{align*}
\end{linenomath}

Both the time for computing $D(n,d,h,p)$ and its overall space are asymptotically the same as those used for $M(n,A,h,p)$. We have thus proven the following:
\begin{corollary}
	Table $D(n,d,h,p)$ and its marginals $D(n,d,h)$ and $D(n,d)$ can be computed for all $\mathcal{O}(n^5)$ many possible parameter values in $\mathcal{O}(n^7)$ total time. 
\end{corollary}

\end{document}